\documentclass[journal,12pt,onecolumn, draftclsnofoot,]{IEEEtran}

\usepackage{subfigure}
\usepackage{graphicx}  
\usepackage{dcolumn}   
\usepackage{color}
\usepackage{bm}        

\usepackage{amsmath}
\usepackage{amsfonts}
\usepackage{amssymb}
\usepackage{amsthm}
\usepackage{dsfont}

\usepackage{array}
\usepackage{url}

\usepackage[T1]{fontenc}
\usepackage[utf8]{inputenc}
\usepackage{authblk}

\usepackage{scalerel,stackengine}
\stackMath
\newcommand\reallywidehat[1]{%
\savestack{\tmpbox}{\stretchto{%
  \scaleto{%
    \scalerel*[\widthof{\ensuremath{#1}}]{\kern-.6pt\bigwedge\kern-.6pt}%
    {\rule[-\textheight/2]{1ex}{\textheight}}
  }{\textheight}%
}{0.5ex}}%
\stackon[1pt]{#1}{\tmpbox}%
}

\usepackage{tikz}
\usetikzlibrary{automata,arrows,calc,arrows,decorations.pathmorphing,3d,decorations.footprints,fadings,calc,trees,mindmap,shadows,decorations.text,patterns,positioning,shapes,matrix,fit}

\usepackage[dvipdfm]{hyperref}

\newtheorem{theorem}{Theorem}

\newcommand{\para}[1]{\smallskip\noindent {\bf #1}~}
\newcommand{\ie}{{\em i.e.}}

\newlength{\figurewidthA} 
\newlength{\figurewidthB} 
\newlength{\figurewidthC} 
\newlength{\figurewidthD} 
\newlength{\figurewidthE} 
\title{A Unified Framework for Information Consumption Based on Markov Chains}

\author[1]{David Shui Wing Hui}
\author[1]{Yi-Chao Chen}
\author[1]{Gong Zhang}
\author[1]{Weijie Wu}
\author[2]{\par Guanrong Chen}
\author[3]{John C. S. Lui}
\author[1]{Yingtao Li}
\affil[1]{Huawei Technologies Co. Ltd., China}
\affil[2]{City University of Hong Kong, Hong Kong SAR, China}
\affil[3]{The Chinese University of Hong Kong, Hong Kong SAR, China}

\date{November 25, 2015}

\begin{document}
\maketitle

\begin{abstract}
This paper establishes a Markov chain model as a unified framework for understanding information consumption processes in complex networks, with clear implications to the Internet and big-data technologies. In particular, the proposed model is the first one to address the formation mechanism of the "trichotomy" in observed probability density functions from empirical data of various social and technical networks. Both simulation and experimental results demonstrate a good match of the proposed model with real datasets, showing its superiority
over the classical power-law models.
\end{abstract} 

\section{Introduction} \label{sec:introduction}

Many complex network models have been proposed to provide an
essential macroscopic understanding of various complex real-world networks, such as the Internet and WWW \cite{barabasi:science:random},
metabolic networks \cite{jeong:nature:metabolic}, the ecosystem
\cite{martin:pnas:powerlaw}, as well as citation
\cite{price:jasis:citation} and co-authorship \cite{martin:pre:citation} networks. One important feature of such networks lies in the node-degree distribution. There are two major classes
of node-degree distributions in complex networks. One is the Poisson distribution or exponential distribution (mainly for homogeneous networks, with rapidly decaying tails in the distributions). The other is the power-law distribution (mainly for heterogeneous networks, well known for their scale-free properties, with long tails in the distributions).
Existing models typically account for the occurrence of Poisson \cite{ErdoRenyi:pm:rg} or exponential distribution \cite{callaway:pre:rgg} by the random attachment mechanism \cite{liu:pla:pa}
during the network formation process, while the occurrence of
power-law distribution comes from the preferential attachment
mechanism \cite{liu:pla:pa, yule:ptrsb:math}. Although both mechanisms are essential in network formation and can capture many real-world phenomena to a certain extent, typically each of them works only within a particular range of the broad degree-distribution spectrum. In this work (see Section~\ref{sec:exp}), nine real datasets are analyzed, ranging from citation networks and social networks, to vehicular networks, where all of them exhibit the "trichotomy" phenomenon --- a power-law distribution in the middle range of the distribution with exponential-alike distributions in both the head and the tail regions. In other words, these two typical distributions, although fundamental, are not sufficient to individually provide an accurate fit to the entire range from degree distribution for many real datasets.

Due to the discrepancy in empirical data and in the above-discussed two basic distributions for describing real datasets, several extensions had been suggested based upon the classical preference attachment mechanism, by introducing a fitness value (a weighting factor)
to the preference, called the fitness model
\cite{bianconi:prl:complexnetworks}, or including nonlinear
dependence (on the node-degrees \cite{krapivsky:prl:random} or
node fitness \cite{bianconi:epl:evolving_networks}), or decaying
fitness accounting for temporal effects of the fitness value
\cite{medo:prl:powerlaw_decay}, as well as employing a rewiring mechanism
\cite{albert:prl:evolving_networks} during the network formation
process. Extending the original Barab\'asi-Albert (BA) model \cite{barabasi:science:random},
these variants aim at explaining several extra features in the
formation process of a network, e.g. reflecting the heterogeneous
properties from nodes
\cite{bianconi:prl:complexnetworks, medo:prl:powerlaw_decay}
(using a higher node-specific fitness value to represent a higher attractiveness of the node in 
\cite{bianconi:prl:complexnetworks}), investigating the time-dependence of the connection mechanism from different
perspectives (incorporating the nonlinear dependence in
\cite{krapivsky:prl:random}, or introducing a decaying factor
into the fitness model in \cite{medo:prl:powerlaw_decay}), or
finding a better rewiring mechanism in fitting the power-law 
exponent (i.e. the slope of the degree distribution curve
in log-log scale) in \cite{albert:prl:evolving_networks}.

Despite the fact that the aforementioned models have improved the fitting to empirical data, in many networks (particularly evolving networks), there is a clear occurrence of phase transition points in the degree distribution, which has not been investigated previously in the literature. In particular, a unified model of complex-network
generation mechanisms has not been established to combine both the basic
BA model (having a power-law distribution) \cite{barabasi:science:random}, and the Erd\"{o}s-R\'{e}nyi
(ER) model (having a Poisson distribution) \cite{ErdoRenyi:pm:rg} or its variant on random graphs (having exponential distributions) \cite{callaway:pre:rgg}, with an ability
of explaining the commonly-observed trichotomy in the degree
distribution curves of the empirical data.

This paper proposes a new framework (as an alternative to the various fitness models) for complex network formation processes, which can provide better matching with real datasets. In particular, this paper examines complex networks driven by the next-generation
Internet development, with emphasis on how the information contents in such networks are being processed (abbreviated as \textit{information consumption} processing)
such as in the vehicle-to-vehicle (V2V) network, online social networks
(e.g., Facebook, Twitter friendship networks), and citation networks (or coauthorship networks). The proposed model by nature suggests a clean-slate
approach based on Markov chains, which is proved in this paper to be more powerful and transparent to explain the trichotomy feature in degree distributions arising from many real-world network datasets.

\section{The Proposed Model} \label{sec:model}

In this section, a unified model is proposed based on Markov chains, referred to as the {\it MC model} hereafter, for representing the information consumption process in a complex network.

\subsection{The MC model} \label{subsec-model}

\subsubsection{Model description} \label{subsubsec:model-Description}

The information consumption process is modelled as an evolving network. Every information consumption attempt is represented by a node. The network starts with a few connected nodes and then the network evolves as follows:

\para{Big-Bang (Start):} A node comes to the network randomly at a certain rate; upon arrival, it creates one (or several) connection(s) to existing nodes, like a small big-bang in the network.

In the network, existing nodes have different degrees. As the network evolves, their degrees can change, particularly the evolvement of a particular existing node is divided into three phases --- initializing phase, fast-evolving phase and maturing phase.

\para{Initializing: }
All nodes in this phase are being attached by the new node presumably with an \rm{equal probability}. Hence, the evolution of information consumption is independent of the existing network structure and behavioural pattern. In other words, there is no preferential attachment mechanism in this stage.

\para{Fast-evolving: }
Among existing nodes in this phase, when a new node comes, it prefers connecting to an existing node of higher degree with a higher probability.
Specifically, the connecting probability is proportional to the degree of the existing node. As a result, from the perspectives of those existing nodes with higher degrees (which are usually those nodes in the network with longer residential time), new nodes join the network at a faster rate than in the initializing phase.
In this phase, the evolution of information consumption depends on the existing network structure.
This preferential attachment mechanism enables some nodes to become more and more popular as the network evolves.

\para{Maturing (Saturating): }
For any existing node in this phase, once they reach sufficiently large degree, the attractiveness of these nodes for being attached begins to saturate due to physical, economical or technological (especially computational complexity, in particular due to the practical bound on the implementation complexity of the referential attachment mechanism) constraints. These nodes with degrees exceeding a certain threshold are called super nodes. Although new nodes still preferentially join the super nodes but they would join these nodes with the same probability. There are two reasons. One is that existing super nodes could start to refuse new connections due to physical or economical constraints. The other is that although there are still some differences in the degrees among super nodes, new attempts will not care or cannot distinguish the exact degrees of those super nodes to perform an exact preferential attachment. This is due to the practical bound of the implementation complexity, so that the degree counting of these super nodes will stop after reaching a threshold, and these super nodes having degrees beyond this threshold are indistinguishable to new comers. In both cases, a new attempt will attach to the super nodes with (approximately) the same probability. This means that in the maturing phase, the evolution of the information consumption pattern is still dependent on the existing information consumption pattern, but the dependence is becoming weaker as the network is getting larger (and hence, more severely saturated). This bounded preferential attachment mechanism, in practice, describes the realistic situation where it is no longer meaningful for a new attempt to connect to super nodes with different preferences because they are similar from the viewpoint of the new comer.

\subsubsection{Model formulation}\label{subsubsec:model-Math}

To formulate the network model, the simplest possible initial
network is a chain with only two nodes connected by one edge. Starting from this network, each time a new node arrives at rate $\lambda_0$. After that, a new node would connect to an existing one with a probability proportional to its degree subjected to bounds, and hence existing nodes in the network evolve according to the aforementioned three developing phases, in the following manner. 

The new node will have the same connection probability to all existing nodes in the initializing phase. It will then have an increasing connection probability to a node with a higher degree (i.e. preferential attachment) among all nodes in the fast-evolving phase. Finally it will have a bounded preferential attachment probability to nodes in the maturing phase. 

This formation process is mathematically modeled as follows: when an existing node is in the initializing phase, its degree is less than or equal to a lower boundary value ${\mathcal L}$, and when it is in the maturing phase, its degree is larger than an upper boundary value ${\mathcal U}$. Suppose nodes are labeled in ascending order according to their arrival time. When the network size (number of nodes) is $n-1$, and the $n$-th node arrives, it will connect to the $i$-th existing node for $i \in \left\{1, \cdots n - 1 \right\}$ in the network, with an attachment probability proportional to a constant ${\hat k}_i^{(n-1)}$ associated with $i$, referred to as the modified degree. This ${\hat k}_i^{(n-1)}$ is the same as the node degree $k_i^{(n-1)}$ during the fast-evolving phase but it has a lower bound $L$ ($L\le{\mathcal L}$) in the initializing phase and an upper bound $U$ (${\mathcal U}\le U$) in the maturing phase. Denote $\widehat{k_i^{(n)}}$ by $\widehat{k_i}$ and ${k_i^{(n)}}$ by $k_i$, whenever the corresponding network size $n$ is specified. Suppose that an existing node is to be attached starting from degree $k^0$ (for most cases in this study, $k^0 = 1$ or $k^0 = 0$). Then, the modified degree is mathematically defined as follows (with $\widehat{k^0 - 1} \triangleq 0$): 
\begin{equation} \label{def:modified-deg}
{\hat k}_i =
\begin{cases}

L, & {{\text{if }} k^0 \leq {k_i} \leq \mathcal{L}} \\
k_i, & {{\text{if }}\mathcal{L} < {k_i} \leq \mathcal{U}} \\
U, & {{\text{if }} \mathcal{U} < {k_i} \leq N_{\mathcal{T}} },
\end{cases}
\end{equation}
where $N_{\mathcal{T}}$ is the number of nodes in the system at the current observation time $\mathcal{T}$.

From the perspective of a specific node (denoted as node * hereafter),
its degree $k_{*}$ and the corresponding network size $n$ together form a two-dimensional (continuous-time) Markov chain, with its state
transition rate diagram depicted by Fig.~\ref{fig:model-2DMC} (with normalized rate $\lambda \triangleq {\lambda_0}/{L}$).

\begin{figure}[htbp!]
	\centering
  	\resizebox{\textwidth}{!}{%
  	\begin{tikzpicture}[->, >=stealth', auto, semithick, node distance=5cm]
\tikzstyle{every state}=[fill=white,draw=black,thick,text=black,scale=1, minimum width={width("k-1, n-1")+70pt}]
\node[state]    (A)                     {\LARGE $k, n-1$};
\node[state]    (B)[below of=A]   {\LARGE $k-1, n-1$};
\node[state]    (C)[above left of=B, draw=white]   {\LARGE $\dots$};
\node[state]    (D)[right of=A, xshift=1.5cm]   {\LARGE $k, n$};
\node[state]    (E)[below right of=D, draw=white]   {\LARGE $\dots$};
\node[state]    (F)[right of=D, xshift=3cm]   {\LARGE $k, n+1$};
\node[state]    (G)[above of=F]   {\LARGE $k+1, n+1$};
\path
(A) edge[bend left,above]      node{\LARGE $\lambda \sum\limits_{i \neq *} \widehat{k}_i^{(n-1)}$}      (D)
(B) edge[right]     node{\LARGE $\lambda \left(\widehat{k - 1} \right)$}           (D)
(D) edge[bend left,below]      node{\LARGE $\lambda \sum\limits_{i \neq *} \widehat{k}_i^{(n)}$}      (F)
(D) edge[right]     node{\LARGE $\lambda \left(\widehat{k} \right)$}           (G);

\draw[red] ($(A)+(0,-2)$) ellipse (3cm and 6cm)node[yshift=5.5cm]{\LARGE When network size is $n - 1$};
\draw[blue] ($(D)+(5,0)$) ellipse (15cm and 2cm)node[xshift=9.5cm]{\LARGE When degree of Node * is $k$};

\end{tikzpicture}
}
        \vspace{-15pt}
  		\caption{The state transition rate diagram of a two-dimensional Markov chain for a specific node * --- where $k$ denotes its degree and $n$ denotes the size of network that this node * is in.}
  		\label{fig:model-2DMC}
  		\vspace{-15pt}
	\end{figure}
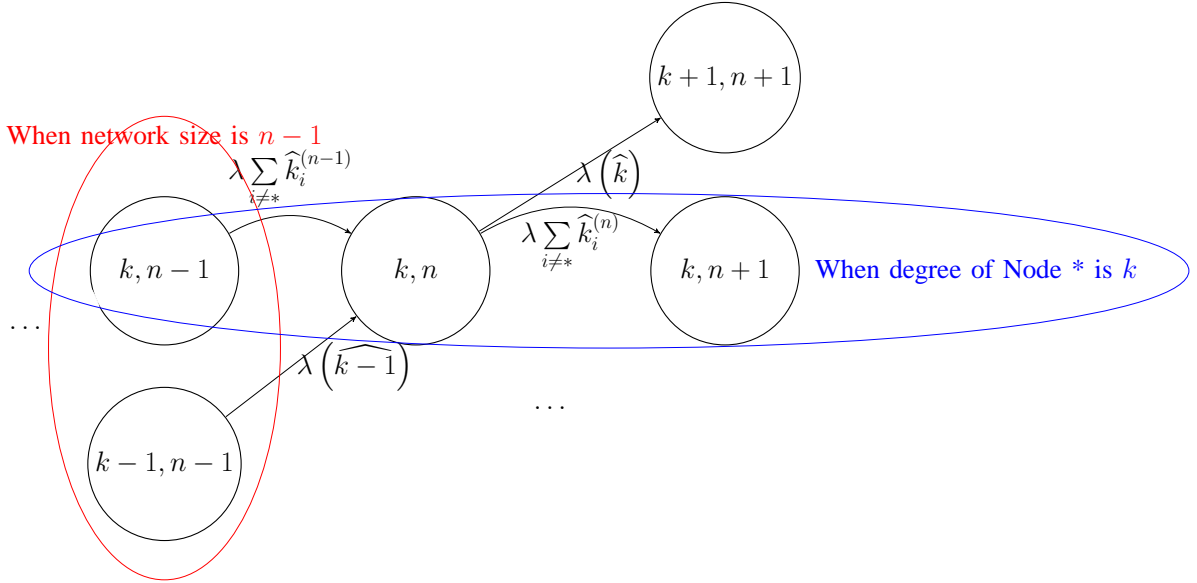


Now, the MC model is further described and discussed.


In Fig. \ref{fig:model-2DMC},
a circle represents a state of the state variable. Suppose that node * in
consideration has degree $k-1$ and the current network size is $n-1$. Then, according to the bounded preferential attachment mechanism, a newly arrived node connects to node * at rate $\lambda$. But, it is also possible for this new node to connect to other nodes. Since the new node is likewise preferentially attached to other nodes in proportion to their degrees, the total state transition rate of connecting to other nodes is given by $\lambda \sum_{i \ne *} {{{\hat k}_i}}$. In the following, consider that it connects to node $i'$.

Denote the number of nodes in the network at time $t$ by $N(t)$ and denote the degree of the specified node * at time $t$ by $K(t)$ as two random processes. Then, the joint probability mass (density) function of
this node degree and the network size at any particular time $t$ is ${p_{k,n}}(t) = \Pr \left\{ {K(t) = k,N(t) = n}\right\}$.
Also denote the sum of the modified degrees of all nodes in the network by $S_n$. Then, when the network has size
$n-1$, one has ${S_{n - 1}} = \sum\limits_{i = 1}^{n - 1}
{{{\hat k}_i^{(n-1)}}}$, and thus ${S_n} = \sum\limits_{i =1 \hfill\atop i
\ne i'\hfill}^{n - 1} {{{\hat k}_i^{(n-1)}}}  + \left(\widehat{{k_{i'}^{(n-1)}} + 1}\right)
+ {\hat k_n^{(n)}}$.

As can be verified by examining the state transition rate diagram in Fig. \ref{fig:model-2DMC}, the dynamics of the probability mass function of the state variable --- the degree $k$ of a specific node * (for $k \in \left\{1, \cdots N_\mathcal{T} \right\}$) and the network of size $n$ (for $n \in \left\{1, \cdots N_\mathcal{T} \right\}$) satisfy the following equation:

\begin{eqnarray}
\frac{d}{{dt}}{p_{k,n}}(t)
&=& \lambda \left( {\widehat {k - 1}}
\right){p_{k - 1,n - 1}}(t) + \lambda \left( {\sum_{i \ne *}
{{{\hat k}_i}} } \right){p_{k,n - 1}}(t)
 - \lambda {S_n}{p_{k,n}}(t) \nonumber\\
&=& \lambda \left( {\widehat {k - 1}} \right){p_{k - 1,n - 1}}(t)
 + \lambda \left( {{S_{n - 1}} - \hat k} \right){p_{k,n - 1}}(t)
 - \lambda {S_n}{p_{k,n}}(t), \label{eqn:DE-on-2Dpmf}
\end{eqnarray}
with the boundary conditions $p_{k, n_* + k - 2}(t) = 0$ for all $k \ge 1$ (assuming that when node * arrives, the network contains $n_*$ of nodes) and $p_{k, N_\mathcal{T}+d}(t) = 0$ for all $k \ge 1$ and $d \ge 1$.

Summing equation (\ref{eqn:DE-on-2Dpmf}) over all the possible network size $n$ from $2$ to $N_{\mathcal{T}}$,  one could obtain the dynamics of the degree $k$ (for $k \in \left\{1, \cdots N_\mathcal{T} \right\}$) of the specified node * as follows:
\begin{equation}
\frac{d}{{dt}}{p_{k}}(t) = \lambda \left( {\widehat {k - 1}} \right){p_{k - 1}}(t)
- \lambda \left( {\hat k} \right){p_{k}}(t) + \lambda \left(\hat{k}\right) {p_{k, N_{\mathcal{T}}}}(t) \label{eqn:DE-on-pmf-implicit-w-BdyEffect}
\end{equation}
for simplicity of analysis, one could remove the small boundary term $\lambda \hat{k} {p_{k, N_{\mathcal{T}}}}(t)$ (which vanishes as $\mathcal{T} \to \infty$), as a result, the degree dynamics is simplified as follows:
\begin{equation}
\frac{d}{{dt}}{p_{k}}(t) = \lambda \left( {\widehat {k - 1}} \right){p_{k - 1}}(t)
- \lambda \left( {\widehat k} \right){p_{k}}(t), \quad \forall k \ge k^0
\label{eqn:DE-on-pmf-implicit}
\end{equation}
where $k^0$ is the starting degree, which is either 0 or 1,

Write it in a more explicit manner, one has
\begin{equation} \label{eqn:DE-on-pmf}
\frac{d}{{dt}}{p_k}(t) = \left\{ {\begin{array}{*{20}{c}}
{ - \lambda L{p_k}(t)},\\
{\lambda L{p_{k - 1}}(t) - \lambda L{p_k}(t)},\\
{\lambda \left( L \right){p_{k - 1}}(t) - \lambda k{p_k}(t)},\\
{\lambda \left( {k - 1} \right){p_{k - 1}}(t) - \lambda k{p_k}(t)},\\
{\lambda \left( {k - 1} \right){p_{k - 1}}(t) - \lambda U{p_k}(t)},\\
{\lambda U{p_{k - 1}}(t) - \lambda U{p_k}(t)},
\end{array}} \right.
\begin{array}{*{20}{l}}
{{\text{if }} k = {k^0}}\\
{{\text{if }} k^0 < k \le \mathcal{L}}\\
{{\text{if }} k = \mathcal{L} + 1}\\
{{\text{if }} \mathcal{L} + 1 < k \le \mathcal{U}}\\
{{\text{if }} k = \mathcal{U} + 1}\\
{{\text{if }} \mathcal{U} + 1< k \le N_{\mathcal{T}}}
\end{array}
\end{equation}
which could be further simplified if $L = \mathcal{L}$ or $U = \mathcal{U}$.

It is noted that when deriving equation (\ref{eqn:DE-on-pmf-implicit}), a boundary term $\lambda {\hat{k}}{p_{k,{N_{\mathcal{T}}}}}(t)$ in (\ref{eqn:DE-on-pmf-implicit-w-BdyEffect}) is omitted, as it is very small as compared to other terms in equation (\ref{eqn:DE-on-pmf-implicit-w-BdyEffect}). Actually, in $p_k(t) = \sum\limits_{n = 2}^{N_{\mathcal{T}}} p_{k,n}(t)$, the term ${p_{k,{N_{\mathcal{T}}}}}(t)$ is the smallest term in the sum, and it vanishes as $\mathcal{T} \to \infty$. Hence, neglecting this boundary term would not cause significant difference in equation (\ref{eqn:DE-on-pmf-implicit}), but can greatly simplify the subsequent analysis. 

Pictorially, the process of summing equation (\ref{eqn:DE-on-2Dpmf}) over $n$ means that the above two-dimensional MC (in Fig. \ref{fig:model-2DMC}) is reduced to the following one-dimensional MC in the degree variable (with respect to the specified node *), which represents the degree dynamics, together with the consideration of big-bang phase, the degree dynamics is shown in Fig. \ref{fig:model-MC}, which is a birth process (differing from the Yule process \cite{yule:ptrsb:math} by having a lower bound and an upper bound on the birth rates) modified with a big-bang start (represented in the figure by the beginning transitions, which are parametrized by $p_1^0, p_2^0, \ldots, p_{\mathcal{L} + 1}^0$, corresponding to the probabilities of initiating $1, 2, \ldots,\mathcal{L} + 1$ connections due to the arrival of new nodes, respectively). In practice, these probabilities of the number of connections that each new attempt makes are network-specific parameters. Denote the probability mass function of the degree of the specified node *
by ${p_k}(t)= \Pr \left\{ {K(t) = k} \right\}$. The state transition rate diagram of its degree shown in Fig. \ref{fig:model-MC} has two interpretations depending on how each node is accounted in the degree distribution of the network. For accounting a node only when it arrives(i.e. $k^0 = 1$), Fig. \ref{fig:model-MC} represents the state transition rate diagram of its node degree, with initial probabilities of the state, $p_1\left(0\right) = p_1^0, p_2\left(0\right) = p_2^0, \ldots, p_{\mathcal{L} + 1}\left(0\right) = p_{\mathcal{L} + 1}^0$ (and the new node arrives at rate $\lambda L$), assuming that the node arrives at time 0 without loss of generality. For accounting a node even when it is isolated (i.e. $k^0 = 0$), Fig. \ref{fig:model-MC} represents the state transition rate diagram of its node degree (even when it has not yet arrived, i.e. with degree 0), with initial probabilities of the state, $p_0\left(0\right) = 1$ and $ p_k\left(0\right) = 0$, for all other $k \neq 0$.
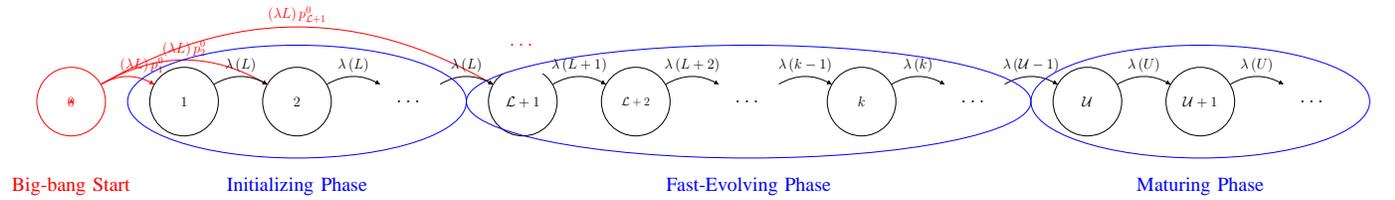
\begin{figure}[htbp!]
  \centering
  \vspace{-10pt}
  \resizebox{1.1\textwidth}{!}{%
  \begin{tikzpicture}[->, >=stealth', auto, semithick, node distance=4cm]
\tikzstyle{every state}=[fill=white,draw=black,thick,text=black,scale=1, minimum width={width("k-1, n-1")+30pt}]
\node[state]    (0)[draw=red, text=red]   {\large $0$};
\node[state]    (1)[right of=0]    {\large $1$};
\node[state]    (2)[right of=1]    {\large $2$};
\node[state]    (3)[right of=2,draw=white]    {\LARGE $\dots$};
\node[state]    (A)[right of=3]    {\large $\mathcal{L} + 1$};
\node[state][yshift=-2cm]    (B)[above of=A, draw=white,text=red]    {\LARGE $\dots$};
\node[state]    (C)[right of=A]   {$\mathcal{L} + 2$};
\node[state]    (D)[right of=C, draw=white]   {\LARGE $\dots$};
\node[state]    (E)[right of=D]   {\large $k$};
\node[state]    (F)[right of=E, draw=white]   {\LARGE $\dots$};
\node[state]    (U)[right of=F]   {\large $\mathcal{U}$};
\node[state]    (U+1)[right of=U]   {\large $\mathcal{U} + 1$};
\node[state]    (Final)[right of=U+1, draw=white]   {\LARGE $\dots$};
\path
(0) edge[bend left,above,red]      node[xshift=0.5cm] {\large $\left(\lambda L\right) p_1^0$}      (1)
(0) edge[bend left,above,red]      node{\large $\left(\lambda L\right) p_2^0$}      (2)
(0) edge[bend left,above,red]      node{\large $\left(\lambda L\right) p_{\mathcal{L} + 1}^0$}      (A)
(1) edge[bend left,above]      node{\large $\lambda \left(L \right)$}      (2)
(2) edge[bend left,above]      node{\large $\lambda \left(L \right)$}      (3)
(3) edge[bend left,above]      node{\large $\lambda \left(L \right)$}      (A)
(A) edge[bend left,above]      node{\large $\lambda \left(L + 1 \right)$}      (C)
(C) edge[bend left,above]      node{\large $\lambda \left(L + 2 \right)$}      (D)
(D) edge[bend left,above]      node{\large $\lambda \left(k - 1\right)$}      (E)
(E) edge[bend left,above]      node{\large $\lambda \left(k\right)$}      (F)
(F) edge[bend left,above]      node{\large $\lambda \left(\mathcal{U} - 1\right)$}      (U)
(U) edge[bend left,above]      node{\large $\lambda \left(U\right)$}      (U+1)
(U+1) edge[bend left,above]      node{\large $\lambda \left(U\right)$}      (Final);

\draw[red] ($(0)+(0,0)$) node[yshift=-3cm]{\LARGE Big-bang Start};
\draw[blue] ($(2)+(0,0)$) ellipse (6cm and 2cm)node[yshift=-3cm]{\LARGE Initializing Phase};
\draw[blue] ($(D)+(0,0)$) ellipse (10cm and 2cm)node[yshift=-3cm]{\LARGE Fast-Evolving Phase};
\draw[blue] ($(U+1)+(0,0)$) ellipse (6cm and 2cm)node[yshift=-3cm]{\LARGE Maturing Phase};

\end{tikzpicture}
}
  \vspace{-30pt}
  \caption{The state transition rate diagram of the MC representing the development of degree $k$ of each node of the network. The \textcolor{blue}{blue} part represents the state transition rate diagram of the Markov Chain with degree $k$ of each arrived node as the state variable. The \textcolor{red}{red} part represents new node coming at a rate of $\lambda L$, with the probability $p_i^0$ of making $i$ initial connections with existing nodes of the network. It leads the initial distribution of the degree (i.e. the states of the MC) of the arrived node to be $\Pr\left(\text{initial degree} = i\right) = p_i^0$.}
  \vspace{-10pt}
  \label{fig:model-MC}
\end{figure}

Now, one is ready to analyze the node degree distribution $p_k$. Beside deriving $p_k\left(t\right)$ from equation (\ref{eqn:DE-on-pmf}), one also needs to know how long (i.e. the residential-time, $T$) the specified node * has been staying in the network. Suppose that the network starts at time 0, and node * arrives at time $t_*$. Then, at an observation time $\mathcal{T}$, the residential-time $T$ of node * is $T_* = \mathcal{T} - t_*$. Denote the corresponding distribution as $f_T\left(t\right)$. This distribution depends on the differential-difference equation for the marginal
distribution of $N(t)$, which could be obtained by summing equation
(\ref{eqn:DE-on-2Dpmf}) over $k$. The detailed analysis and calculations of these distributions are presented in the Appendix.

\subsubsection{Comparison with existing models}
\label{subsubsec:existing-compare}

In this section, the new physical meaning addressed in the proposed MC model compared with existing models are discussed. Recall the MC model is proposed to represent the information consumption process on a complex network. Under this model, the network starts with a big-bang and evolves according to three phases of development — initializing phase, fast-evolving phase and maturing phase. In retrospect, the fast-evolving phenomenon in different networks is also reported in existing models, such as the BA model and its variants mentioned in Section \ref{sec:introduction} (the introduction section). When these BA-type models are used to describe information consumption patterns, and cast into the proposed MC model, those models have only one phase --- the fast-evolving phase, corresponding to the preferential attachment mechanism based only on the degrees of the existing nodes but not related to the size or other constraints of the current network structure. 

The major differences of this proposed MC model, as compared to existing models, are on the initializing phase and the maturing phase, both representing physical constraints but in different contexts - the initializing phase for characterizing the constraint due to system set-up time, and the maturing phase for characterizing the constraint due to physical, economical and technological limitations. These extra features in the proposed MC model provide some justifications for the observed trichotomy in density functions of various networks in Section \ref{sec:exp} (the experiment section), particularly next generation information networks, which was not addressed in existing literature. Specifically, it will be shown that most real-world networks have nodes in the initializing phase (particularly at the beginning of the network evolution) and nodes reaching the maturing phase (particularly after the network contains a sufficiently large number of nodes). 

Furthermore, differing from the existing literature, the proposed model includes a big-bang phase, which allows random multiple start capabilities. In particular, it helps in providing a more reasonable final close-form expression in node-degree distribution of the network presented in Theorem \ref{thm:trichotomyLaw} in the Section \ref{sec:theory}, which is closer to the empirical results, as well as the resultant model is more close to real situations. Specifically, in the big-bang phase, each time when a new information consumption attempt arrives, it will connect to one or multiple existing attempts at random. Such a possibility of multiple connections from a new attempt is also needed in the model as several important networks also have a bursty change (instead of incremental change) in the information consumption patterns upon each attempt arrival. Example networks include citation networks and the next generation vehicular networks. In citation networks, when a new paper is published, it will usually cite several (instead of one) papers; in vehicular network, when a new vehicle comes into a communication zone, it will be in contact (or in feasible communication region) with multiple (instead of one) vehicles. 

Mathematically, the BA-type models, from the proposed MC model viewpoint, are birth processes with a rate proportional to the state variable (i.e. a Yule process \cite{yule:ptrsb:math}), while the proposed general MC model is a modified Yule process with a rate accounting both the state variable and the node degree bounds on the connection probabilities due to these physical constraints, and a big-bang start. 

Therefore, the proposed MC
model is more realistic than the BA model in representing the information consumption processes and patterns in general complex networks. 

\section{Table of Notations} \label{sec:notations}

For the ease of reference in the discussions of later sections, a list of important system parameters (in the proposed MC model) is summarized in Table \ref{table:system-parameters}.

\begin{table}[htbp!]
  \centering
  {
  \begin{tabular}{|c|c|}
	\hline 
	Notation & Meaning \\
	\hline 
	$\mathcal{L}$ & Lower boundary value of the 	degree of a node signifying the end of the initializing phase \\ 
	\hline 
	$\mathcal{U}$ & Upper boundary value of the 	degree of a node signifying the start of the maturing phase \\ 
	\hline 
	$L$ & Lower bound of modified degree for lower bounding the attachment probability in the whole initializing phase \\ 
	\hline 
	$U$ & Upper bound of modified degree for upper bounding the attachment probability in the whole maturing phase \\
	\hline 
	$\lambda$ & Normalized arrival rate, i.e., average arrival rate of new attempts divided by $L$\\   
	\hline 
	$\mathcal{T}$ & Current observation time of the system \\ 
	\hline 
	$N_{\mathcal{T}}$ & The number of nodes in the system at time $\mathcal{T}$ \\ 
	\hline 
	\end{tabular}
	}
	\caption{A list of important system parameters}
	\label{table:system-parameters}
\end{table}

As time $t$ increases from time $0$ to the observation time $\mathcal{T}$, the system parameters are fixed, but the system variables could change. A list of important system variables used is also summarized in Table \ref{table:system-variables}.

\begin{table}[htbp!]
  \centering
  {
\begin{tabular}{|c|c|}
\hline 
Notation & Meaning \\ 
\hline 
$k_{*}$ & Degree of a specific node * \\ 
\hline 
$n$ & Network size - i.e., number of nodes in the network \\ 
\hline 
$\hat{k_{*}}$ & Modified degree of a specific node * defined in equation (\ref{def:modified-deg}) \\ 
\hline 
$T_{*}$ & 
Residential-time of a specific node *, i.e., $\mathcal{T} - t_*$, where $t_*$ is the arrival time of node *
\\ 
\hline 
$t$ & 
Running time index of the Markov Chains (M.C.) 
\\
\hline 
\end{tabular} 
	}
	\vspace{10pt}
	\caption{A list of important system variables}
	\label{table:system-variables}
\end{table}

In later discussions, the subscript * used in the notations of the above variables (presented in Table \ref{table:system-variables}) would be omitted, whenever the specified node * is explicitly defined in the text, e.g., $k$ will be used as an abbreviation of $k_*$, and $T$ will be used as an abbreviation of $T_*$. Since the values of the system variables depend on time $t$, this time variable $t$ will be explicitly specified in the text whenever there is a need of clarification, e.g., $p_{k, n}(t)$ is used to denote the probability that "the degree of the specified node * $= k$, and the network size $= n$, at time $t$".

\section{Theoretical Results of the proposed models}
\label{sec:theory}

Mathematically, one can further see that the proposed MC model is a generalization of several existing models such as the Poisson Network model, Exponential Network model, and Power-law Network (e.g. BA) model, by appropriately setting the physical constraint-related parameters: lower bound $L$, lower threshold $\mathcal{L}$, upper bound $U$, and upper threshold $\mathcal{U}$, as proved in the following subsection \ref{subsec:unifying-theorems}.

Furthermore, the proposed MC model can derive new closed-form results on degree distributions which fit much better with the empirical data than the classical models, which is also proved in the following subsection \ref{subsec:unifying-trichotomy}. In particular, this MC model is the first model to explain the observed trichotomy phenomenon. 

\subsection{Generality of the MC model}
\label{subsec:unifying-theorems}

\subsubsection{MC generalizes the Poisson network model and exponential network model}
\label{subsubsec:unifying-Poisson}

\begin{theorem} \label{thm:Poisson}
When $L = \mathcal{L} = \mathcal{U} = U$ in the MC model (where $p_k^0 = 0$ for all $k$ except
$p_1^0 = 1$), as $\mathcal{T} \to \infty$, it reduces to one of the two classical models respectively, for
\begin{itemize}
\item Case 1 --- accounting a large and fixed set of nodes (starting with 0 degree, i.e., $k^0 = 0$):\\ it reduces to the Poisson network model; \item Case 2 --- accounting all nodes, but excluding isolated nodes (i.e., $k^0 = 1$):\\it reduces to the exponential network model.
\end{itemize}
\end{theorem}

\begin{proof}
When $L = \mathcal{L} = \mathcal{U} = U$, from equation (\ref{eqn:DE-on-pmf}), when $\mathcal{T}$ is large enough, the dynamics of the degree distribution satisfy
\begin{equation} \label{eqn:DE-on-pmf-PoiGeom}
\frac{d}{{dt}}{p_k}(t) = \left\{ {\begin{array}{*{20}{c}}
{ - \lambda L{p_k}(t)},\\
{\lambda L{p_{k - 1}}(t) - \lambda L{p_k}(t)},
\end{array}} \right.\begin{array}{*{20}{c}}
{{\text{if }}k = k^0},\\
{{\text{if }} k > k^0}.\\
\end{array}
\end{equation}
By taking the Laplace transform, one has
\begin{eqnarray}
s{P_k}\left( s \right) - {p_k}\left( 0 \right)
&=& \lambda L{P_{k - 1}}\left( s \right)
    - \lambda L{P_k}\left( s \right),\nonumber\\
{P_k}\left( s \right) &=& \frac{{\lambda L}}{{s + \lambda L}}{P_{k - 1}}
    \left( s \right) + \left( {\frac{p_k \left(0\right)}{{s + \lambda L}}}
    \right), \label{eqn:Laplace-Geomtric}
\end{eqnarray}
where ${P_{k^0 - 1}}\left( s \right)$ is defined to be 0.
\begin{itemize}
\item In Case 1, $k^0 = 0$,  $P_{-1}(s)$ is defined to be $0$, $p_k\left(0\right) = 0$ for all $k$ except $p_0\left(0\right) = 1$. Thus, by equation (\ref{eqn:Laplace-Geomtric}), $P_0(s) = \frac{1}{s + \lambda L}$, and by mathematical induction (using equation (\ref{eqn:Laplace-Geomtric})),
\begin{eqnarray}\label{eqn:Laplace-PoissonNetwDeg0}
{P_k}\left( s \right)
&=& \sum\limits_{i = 0}^{k - 1} {\left( {\frac{{\lambda L}}{{s + \lambda L}}}
    \right)^{k - i}}\left( {\frac{p_i \left(0\right)}{{s + \lambda L}}}
    \right)\nonumber\\
&=& {\left( {\lambda L} \right)^{k}}
    {\left( {\frac{1}{{s + \lambda L}}} \right)^{k + 1}}\nonumber\\
&=& \frac{{{{\left( {\lambda L} \right)}^{k}}}}
    {{\left( {k} \right)!}}\frac{{\left( {k} \right)!}}
    {{{{\left( {s + \lambda L} \right)}^{k+ 1}}}}.
\end{eqnarray}
Then, the inverse Laplace transform yields
\begin{eqnarray}\label{eqn:PoissonNetwDeg0-pmf}
{p_k}\left( t \right)
&=& \frac{{{{\left( {\lambda L}\right)}^{k}}}}
    {{\left( {k} \right)!}}{t^{k}}
    {e^{ - \lambda Lt}}\nonumber\\
&=& \frac{{{{\left( {\lambda Lt} \right)}^{k}}
 {e^{ - \lambda Lt}}}}{{\left( {k} \right)!}}.
\end{eqnarray}
The specified set of nodes have been existing in the network since the starting of the network formation process, so their residential time is $T = \mathcal{T}$. Then, the degree distribution of these nodes is given by ${p_k}\left(T \right) = \frac{{{{\left( {\lambda L T} \right)}^{k}}
 {e^{ - \lambda L T}}}}{{\left( {k} \right)!}}$, which is a Poisson distribution with parameter $\lambda L T$. This is known as the Poisson network introduced by Erd\"{o}s and R\'{e}nyi \cite{ErdoRenyi:pm:rg}.
\item In Case 2, $k^0 = 1$,  $P_{0}(s)$ is defined to be $0$, $p_k\left(0\right) = 0$ for all $k$ except $p_1\left(0\right) = 1$. Thus by equation (\ref{eqn:Laplace-Geomtric}), $P_1(s) = \frac{1}{s + \lambda L}$, whence 
\begin{eqnarray}\label{eqn:Laplace-ExpNetwoDeg0}
{P_k}\left( s \right)
&=& \sum\limits_{i = 1}^{k - 1} {\left( {\frac{{\lambda L}}{{s + \lambda L}}}
    \right)^{k - i}}\left( {\frac{p_i \left(0\right)}{{s + \lambda L}}}
    \right)\nonumber\\
&=& {\left( {\lambda L} \right)^{k{\rm{ - 1}}}}{\rm{ }}
    {\left( {\frac{{\rm{1}}}{{s + \lambda L}}} \right)^k}\nonumber\\
&=& \frac{{{{\left( {\lambda L} \right)}^{k- 1}}}}
    {{\left( {k - 1} \right)!}}\frac{{\left( {k - 1} \right)!}}
    {{{{\left( {s + \lambda L} \right)}^k}}}.
\end{eqnarray}

Then, the inverse Laplace transform yields
\begin{eqnarray}\label{eqn:ExpNetwoDeg0-pre-pmf}
{p_k}\left( t \right)
&=& \frac{{{{\left( {\lambda L}\right)}^{k - 1}}}}
    {{\left( {k - 1} \right)!}}{t^{k - 1}}
    {e^{ - \lambda Lt}}\nonumber\\
&=& \frac{{{{\left( {\lambda Lt} \right)}^{k - 1}}
 {e^{ - \lambda Lt}}}}{{\left( {k - 1} \right)!}},
\end{eqnarray}
which is a Poisson distribution with parameter $\lambda L$.

\hspace{5pt}
Note that the above $p_k(t)$ only represents the degree distribution of a node * at time $t$ if it arrives at time $0$. More generally, if a node $i$ arrives at time $a_i$, its degree distribution at the observation time $\mathcal{T}$ is given by $p_k\left(\mathcal{T} - a_i \right)$.

\hspace{5pt}
Denote the residential-time of each node $i$ by $T_i = \mathcal{T} - a_i$. Then, the degree distribution of the network is given by $p_k = \frac{1}{N_{\mathcal{T}}}\sum\limits_{i = 1}^{N_{\mathcal{T}}} \int_{0}^{\mathcal{T}} p_{k_i}\left(t \right) f_{T_i} \left(t \right) dt = {\mathds{E}}_{i \sim U\{1, \cdots, N_{\mathcal{T}}\}} {\mathds{E}}_{T_i} \left[ p_{k_i} \left(t \right)\right]$,
where $N_{\mathcal{T}}$ denotes the total number of nodes in the network at the observation time $\mathcal{T}$, and $f_{T_i} \left(t \right)$ denotes the probability density function of the residential-time $T_i$ of node $i$. The symbol $i \sim U\{1, \cdots, N_{\mathcal{T}} \}$ represents that a node $i$ is uniformly selected among the nodes ranged from $1$ to $N_{\mathcal{T}}$.

\hspace{5pt}
In other words, the degree distribution of the network is given by $p_k = {\mathds{E}}_{T} \left[p_k \left(t \right) \right]$, where $T$ is the random variable denoting the residential-time of a randomly picked node, and $p_k \left(t \right)$ is the degree distribution of this randomly picked node with residential-time $T = t$. Thereafter, a crucial step of computing $p_k$ is to find the distribution of the residential-time $T$. The detailed derivation, when $L=\mathcal{L}=\mathcal{U}=U$, is presented in Case 1 of Appendix \ref{AppendixA:residentialTimeByExtOnly}, and is shown to be exponentially distributed with
parameter $\lambda L$, abbreviated as $\exp\left(\lambda L\right)$.

\hspace{5pt}
As an illustration on the computation of the distribution of the residential-time $T$, consider a special case in the MC model where $L=\mathcal{L}=\mathcal{U}=U=1$. The network dynamics satisfy
\[
\frac{d}{{dt}}{p_n}(t) = \lambda \left( {n -1}
\right){p_{n - 1}}(t) - \lambda n{p_n}(t),
\]
which is obtained by summing equation (\ref{eqn:DE-on-2Dpmf}) over $k$. Note that this is a Yule process, which is shown in \cite{yule:ptrsb:math} that the residential-time of a node has an exponential distribution with rate $\lambda$.

\hspace{5pt}
Back to the general case, after averaging the Poisson density $p_k(t)$ in equation (\ref{eqn:ExpNetwoDeg0-pre-pmf}) with the residential-time distribution $\exp\left(\lambda L\right)$, one obtains
\begin{eqnarray} \label{eqn:ExpNetwoDeg0-pmf}
{p_k} &=& {\mathds{E}_T}\left[ {{p_k}\left( t \right)} \right]
 = \int_0^\infty {{\frac{{{{\left( {\lambda Lt}
 \right)}^{k - 1}}{e^{ - \lambda Lt}}}}{{\left( {k - 1}
 \right)!}}}\lambda L {e^{ - \lambda L t}}dt} \nonumber \\
&=& \lambda L {P_k \left(\lambda L \right)}\nonumber \\
&=& \frac{{{{\left( {\lambda L} \right)}^k}}}
 {{\left( {k - 1} \right)!}}\frac{{\left( {k - 1} \right)!}}
 {{{{\left( {\lambda L + \lambda L} \right)}^k}}}\nonumber\\
&=& \frac{1}{2} \left(\frac{1}{2} \right)^{k-1},
\end{eqnarray}
which is a geometric distribution with parameter $\frac{1}{2}$, known as the node-degree distribution of an exponential network \cite{callaway:pre:rgg}, since the geometric distribution is the discrete version of the exponential distribution. 
\end{itemize}

\end{proof}

As a remark, it can be seen from the above proof that the starting degree $k^0 = 0$ (with a fixed set of nodes) or $k^0 = 1$ (with all nodes) corresponds to the well-known classical Erd\"{o}s-R\'{e}nyi (ER) model and exponential network model, respectively. This paper in particular focuses on information consumption modeling. In typical scenarios, only when the new information consumption attempt starts to initiate a connection, its existence will be aware by the system, and the system will start to account this new attempt which has one connection, so the starting degree $k^0$ of any node in the network is $1$. Therefore, in all subsequent discussions, the starting degree $k^0 = 1$ will be used.

It is remarked that from the proof, one can also derive the probability mass function of the node degrees of the network, under the general case $p_1\left(0\right) \neq 1$, as follows
\begin{eqnarray} \label{eqn:mix-truncGeom-pmf}
p_k &=& \sum\limits_{i = 1}^{k - 1} {p_i \left(0\right)}
 {\left( {\frac{{\lambda L}}{{\lambda L + \lambda L}}}
 \right)^{k - i}}\left( {\frac{\lambda L}{\lambda L +
 \lambda L}} \right)\nonumber\\
&=& \sum\limits_{i = 1}^{k - 1} {p_i \left(0\right)}
 \left( {\frac{1}{2}}
 \right){\left( {\frac{1}{2}} \right)^{k - i}},
\end{eqnarray}
which is a mixture of consecutively truncated geometric distributions with parameter $\frac{1}{2}$, and the mixing
probability (or the $i$-th truncation probability) being
$p_i\left(0\right)$. 

Equation (\ref{eqn:mix-truncGeom-pmf}) looks like the discrete analog of the hyperexponential distribution. For hyperexponential distribution, as a mixture of some exponential distributions, the empirical data generated from this distribution can be well-fitted to the distribution using the well-known Prony method (widely used in the engineering literature, e.g., in power systems \cite{hauer:ps90:prony}). However, there are still some differences between equation (\ref{eqn:mix-truncGeom-pmf}) and the hyperexponential distribution. Precisely, equation (\ref{eqn:mix-truncGeom-pmf}) is a mixture of \textit{truncated} geometric distributions. A simple generalization to the discrete analog of the Prony method for parameter estimation on fitting the empirical data may not be good enough, so a further extension of the method is needed to handle the effect of truncations. Since the detailed design of this possibly optimal fitting procedure still needs further investigation, it is left for future research. In this paper, only a simple heuristic procedure is proposed for fitting data from the mixture of truncated geometric distributions (see Section \ref{sec:fitting} for details).

\subsubsection{MC generalizes the power-law model}
\label{subsubsec:unifying-PowerLaw}

\begin{theorem} \label{thm:PowerLaw}
When $L = \mathcal{L} = 1$ and $U = \mathcal{U} = \infty$ in the MC model (where $p_k(0) = 0$
for all $k$ except $p_1(0) = 1$, and  excluding isolated nodes), as $\mathcal{T} \to \infty$, 
it reduces to the power-law model.
\end{theorem}

\begin{proof}
Recall the dynamics of the degree distribution of a node satisfy equation (\ref{eqn:DE-on-pmf-implicit-w-BdyEffect}). As $\mathcal{T} \to \infty$, it can be approximated by equation (\ref{eqn:DE-on-pmf-implicit}). Furthermore, since $L = \mathcal{L} = 1$ and $U = \mathcal{U} = \infty$, the dynamics of the degree distribution can be further simplified as follows:
\[
\frac{d}{{dt}}{p_k}(t) = \lambda \left( {k - 1}
\right){p_{k - 1}}(t) - \lambda k{p_k}(t).
\]

Its Laplace transform is
\[
s{P_k}\left( s \right) - {p_k}\left( 0 \right) =
\lambda \left( {k - 1} \right){P_{k - 1}}\left( s \right)
- \lambda k{P_k}\left( s \right).
\]
Thus,
\begin{eqnarray} \label{eqn:Laplace-PowerLaw}
{P_k}\left( s \right)
&=& \left( {\frac{{\lambda \left( {k - 1}
 \right)}}{{s + \lambda k}}} \right){P_{k - 1}}\left( s \right)
 + \left( {\frac{p_k \left(0\right)}{{s + \lambda k}}}
 \right)\nonumber\\
&=& \sum\limits_{i = 1}^{k} \left( \prod\limits_{j = 1}^{k - i}
 \left(\frac{{\lambda \left(k - j\right)}}{{s + \lambda
 \left( {\left(k - j\right) + 1} \right)}}\right)\right)
 \left( {\frac{p_i \left(0\right)}{{s + \lambda i}}}
 \right)\nonumber\\
&=& \frac{{{\lambda ^{k - 1}}\left( {k - 1} \right)!}}
 {{\prod\limits_{i = 2}^k {(s + i\lambda )} }}
 \left( {\frac{p_1 \left(0\right)}{{s + \lambda }}}
 \right)\nonumber\\
&=& \frac{{{\lambda ^{k - 1}}\left( {k - 1} \right)!}}
 {{\prod\limits_{i = 1}^k {(s + i\lambda )} }},
\end{eqnarray}
where $P_0\left(s\right)$ is defined to be $0$, as $P_1\left(s\right)$
is obtained according to equation (\ref{eqn:DE-on-pmf}), and by assumption $p_k(0) = 0$ for all $k$ except $p_1(0) = 1$.

The detailed derivation of the residential-time of any specific node *, when $L = \mathcal{L}$ and $U = \mathcal{U} = \infty$, is presented in Case 2 of Appendix \ref{AppendixA:residentialTimeByExtOnly}, and is shown to be exponentially distributed with rate $2 \lambda$ as $\mathcal{T} \to \infty$. After averaging the time-dependent node-degree distribution $p_k\left( t\right)$ with this residential-time distribution, one obtains a power law as follows:
\begin{eqnarray}\label{eqn:PowerLaw-pmf}
{p_k} &=& {E_T}\left[ {{p_k}\left( t \right)} \right]
 = \int_0^\infty  {{e^{ - \lambda t}}{{\left( {1 - {e^{ - \lambda t}}}
 \right)}^{k - 1}}2\lambda {e^{ - 2\lambda t}}dt}\nonumber \\
&=& 2\lambda {P_K}\left( 2\lambda \right)\nonumber \\
&=& \frac{{{2 \lambda ^k}\left( {k - 1} \right)!}}
 {{\prod\limits_{i = 1}^k {( 2\lambda + i\lambda )} }}\nonumber \\
&=& \frac{{2\left( {2!} \right)\left( {k - 1} \right)!}}
 {{\left( {k + 2} \right)!}}\nonumber \\
&=& \frac{4}{{k\left( {k + 1} \right)\left( {k + 2}\right)}}
 \sim {k^{ - 3}}.
\end{eqnarray}
\end{proof}

It is remarked that one could follow the same procedure of the proof starting from equation (\ref{eqn:Laplace-PowerLaw}) to obtain the probability mass function of the node degrees of the network, for the general case of $p_1\left(0\right) \neq 1$, as follows:
\begin{eqnarray}
p_k &=& \sum\limits_{i = 1}^{k} {p_i \left(0\right)}
 \left( \prod\limits_{j = 1}^{k - i}
 \left(\frac{{\lambda \left(k - j\right)}}{{ 2 \lambda
 + \lambda \left( {\left(k - j\right) + 1} \right)}}\right)\right)
 \left( {\frac{2 \lambda}{{2 \lambda + \lambda i}}} \right) \nonumber\\
&=& \sum\limits_{i = 1}^{k} {p_i \left(0\right)}
 \left( \prod\limits_{j = i}^{k - 1}
 \left(\frac{j}{j + 3}\right)\right) \left( {\frac{2}{{2 + i}}} \right),
\end{eqnarray}
which is a mixture of some consecutively truncated power-law distributions having a slope parameter $-3$, , specifically ${\frac{2}{{2 + i}}} \prod\limits_{j = i}^{k - 1} \left(\frac{j}{j + 3}\right)$ (for $i \in \left\{1, \cdots, k \right\}$) with the mixing probability (or the $i$-th truncation probability) on the $i$-th term being $p_i\left(0\right)$. The resultant distribution is still a power law with a slightly smaller slope parameter $-3+\epsilon$, where $\epsilon > 0$ is a very small value.

\subsection{Capability of the MC model}
\label{subsec:unifying-trichotomy}

In Section \ref{sec:exp} below, experimental results on information consumption patterns in different network datasets are presented. Notably, all real networks have three phases in their degree distributions, although the three phases in different networks start at different times and last for different durations. By using the proposed MC model, this paper is the first to demonstrate such phenomena in information consumption processes. In fact, the MC model can offer an analytical closed-form expression of the
degree distribution and is capable of explaining the observed three phases in
empirical degree distributions in real networks.

The results are summarized as follows:

\begin{theorem} \label{thm:trichotomyLaw}
The degree distribution of the MC model with initial condition $p_1 \left(0\right) = 1,\, p_k \left(0 \right) = 0$ for all \text{other} $k$, is given by
\begin{equation} \label{eqn:trichotomyLaw-special}
{p_k} \sim
\begin{cases}
{c \cdot \text{geom}\left( {\frac{\gamma}{\gamma + L}} \right)}, & {\text{if }} 1 \le k \le \mathcal{L} \\
c \cdot \text{power-law with exponent} -\left(\gamma + 1\right), &
{\text{if }} \mathcal{L} < k \le \mathcal{U}  \\
{c \cdot c_{PL\left(\mathcal{U} \right)} \cdot \text{geom}\left( {\frac{\gamma}{\gamma + U}} \right)}, &
{\text{if }} \mathcal{U} < k \le N_{\mathcal{T}}
\end{cases}.
\end{equation}
For the general initial condition on $p_k \left(0 \right)$ for all $k$,
it is given by 
\begin{equation} \label{eqn:trichotomyLaw-general}
{p_k} \sim
\begin{cases}
c \cdot \sum\limits_{i = 0}^k p_i\left(0 \right) \times i{\text{th}}\
\text{truncated-geom}\left( {\frac{\gamma}{\gamma + L}} \right), &
{\text{if }} 1 \le k \le \mathcal{L} \\
c \cdot \text{power-law with exponent} -\left(\gamma + 1\right) - \epsilon, &
{\text{if }} \mathcal{L}
< k \le \mathcal{U}  \\
{c \cdot c_{PL\left(\mathcal{U} \right)} \cdot \sum\limits_{i = \mathcal{U}}^k p_i\left(0 \right) \times i{\text{th}}\
\text{truncated-geom}\left( {\frac{\gamma}{\gamma + U}} \right)}, &
{\text{if }} \mathcal{U} < k \leq N_{\mathcal{T}}
\end{cases}
\end{equation}
where $L \leq \gamma \leq L + 1$. For $U \ll N_{\mathcal{T}}$, one has $\gamma \approx L$, while for $U \sim N_{\mathcal{T}}$, one has $\gamma \approx L + 1$; $c$ is a normalization constant making the
total probability to be 1, i.e., $\sum_k p_k = 1$; $c_{PL\left({\mathcal{U}} \right)}$ is a constant obtained by multiplying $\frac{\mathcal{U}}{\gamma}$ with the value of the power law probability mass function at $k = \mathcal{U}$; $N_{\mathcal{T}}$ is the network size at the current time ${\mathcal{T}}$ and $\epsilon < 1$ is a small constant.
\end{theorem}

\begin{proof}
Only the case with the most interesting initial condition, $p_1 \left(0 \right) = 1, p_k \left(0 \right) = 0$
for all \text{other} $k$, is proved here; the proof for the general case is similar.

Define the $\mathcal{T}$-truncated Laplace transform of $p_k(t)$ as follows:
\begin{equation}
P_k(s) = \int_{0}^{\mathcal{T}} e^{-s t} p_k(t) dt.
\end{equation}

Using the rule of integration-by-parts, one has
\begin{eqnarray}
\int_{0}^{\mathcal{T}} e^{-s t} \frac{d p_k(t)}{dt} dt &=& [e^{-s t} p_k(t)]|_{t=0}^{t=\mathcal{T}} + s \int_{0}^{\mathcal{T}} e^{-s t} p_k(t) dt \nonumber\\
&=& - p_k(0) + e^{-s \mathcal{T}} p_k(\mathcal{T}) + s \int_{0}^{\mathcal{T}} e^{-s t} p_k(t) dt.
\end{eqnarray}
Then, the $\mathcal{T}$-truncated Laplace transform of the differential-difference equation
(\ref{eqn:DE-on-pmf-implicit-w-BdyEffect}) is given by
\[
s{P_k}\left( s \right) - {p_k}\left( 0 \right) + e^{-s \mathcal{T}} p_k(\mathcal{T}) =
\lambda \left( {\widehat {k - 1}} \right){P_{k - 1}}\left( s \right)
- \lambda \left( {\hat {k}} \right){P_k}\left( s \right) + \lambda \left( \hat{k} \right) P_{k, N_\mathcal{T}} \left( s \right), \quad \forall k \ge 1,
\]
where $\widehat{0} \triangleq
0$, $P_{k, N_\mathcal{T}} \left( s \right)$ is defined as the $\mathcal{T}$-truncated Laplace transform of ${p_{k, N_{\mathcal{T}}}}(t)$, i.e. $P_{k, N_\mathcal{T}} \left( s \right) \triangleq \int_{0}^{\mathcal{T}} e^{-s t} {p_{k, N_{\mathcal{T}}}}(t) dt$, which is small for all $k$ since ${p_{k, N_{\mathcal{T}}}}(t)$ is small for the reason mentioned above.

Reorganizing the above equation leads to
\begin{eqnarray} \label{eqn:Laplace-Trichotomy-iteration}
{P_k}\left( s \right) &=& \left( {\frac{{\lambda \left( {\widehat {k - 1}} \right)}}
{{s + \lambda \left( {\hat {k}} \right)}}} \right){P_{k - 1}}\left( s \right) + \frac{p_k(0) + \lambda \hat{k} P_{k, N_\mathcal{T}} \left( s \right)  - e^{-s \mathcal{T}} p_k(\mathcal{T})}{{s + \lambda \left( {\hat {k}} \right)}}\nonumber\\
&=& \sum\limits_{i = 1}^{k} \left( \prod\limits_{j = 1}^{k - i}
 \left(\frac{{\lambda \widehat{\left(k - j\right)}}}{{s + \lambda
 \left( \reallywidehat{\left(k - j\right) + 1} \right)}}\right)\right)
 \left( {\frac{p_i \left(0\right) + \lambda \hat{i} P_{i, N_\mathcal{T}} \left(s \right) - e^{-s \mathcal{T}} p_i(\mathcal{T})}{{s + \lambda i}}} \right)\nonumber\\
&=& \sum\limits_{i = 1}^{k} \left( \prod\limits_{j = i}^{k - 1}
 \left(\frac{{\lambda \widehat{j}}}{{s + \lambda
 \left( \widehat{j + 1} \right)}}\right)\right) \left( {\frac{p_i \left(0\right) + \lambda \hat{i} P_{i, N_\mathcal{T}} \left(s \right) - e^{-s \mathcal{T}} p_i(\mathcal{T})}{{s + \lambda i}}} \right) \nonumber\\
 &=& \sum\limits_{i = 1}^{k} \left( \prod\limits_{j = i}^{k - 1}
 \left(\frac{{\lambda \widehat{j}}}{{s + \lambda
 \left( \widehat{j + 1} \right)}}\right)\right) \left( {\frac{\widetilde{p_i}}{{s + \lambda i}}} \right),
\end{eqnarray}
where
\begin{equation} \label{def:tilde_p}
\widetilde{p_i} \triangleq p_i \left(0\right) + \lambda \hat{i} P_{i, N_\mathcal{T}} \left(s \right) - e^{-s \mathcal{T}} p_i(\mathcal{T}).
\end{equation}
It is also noted that when $i > k -1$, one has $\prod\limits_{j = i}^{k - 1}
 \left(\frac{{\lambda \widehat{j}}}{{s + \lambda
 \left( \widehat{j + 1} \right)}}\right) = 0$.

The remaining task to obtain a closed-form expression for $P_k\left( s\right)$ is to express $\prod\limits_{j = i}^{k - 1}
 \left(\frac{{\lambda \widehat{j}}}{{s + \lambda
 \left( \widehat{j + 1} \right)}}\right)$ in terms of the system parameters explicitly for different cases of $k$, and then investigate this expression for different subcases of the corresponding iteration variable $i$ (for $i = \left\{1, \cdots k - 1\right\}$) in the summation term in equation (\ref{eqn:Laplace-Trichotomy-iteration}).

\begin{itemize}
\item Case 1. $1 \leq k \leq \mathcal{L}$, for all $i = \left\{1, \cdots k - 1\right\}$:
\[
\prod\limits_{j = i}^{k - 1}
 \left(\frac{{\lambda \widehat{j}}}{{s + \lambda
 \left( \widehat{j + 1} \right)}}\right)
 = 
\left( {\frac{{\lambda L}}{{s + \lambda L}}}
    \right)^{k - i}.
\]
\item Case 2. $\mathcal{L} < k \leq \mathcal{U} $, there are two subcases of the expressions depending on $i$:
\[
\prod\limits_{j = i}^{k - 1}
 \left(\frac{{\lambda \widehat{j}}}{{s + \lambda
 \left( \widehat{j + 1} \right)}}\right)
 = 
\left\{ {\begin{array}{*{20}{c}} \left(\frac{L}{\mathcal{L}}\right)
\prod\limits_{j = \mathcal{L}}^{k - 1} \left(\frac{{\lambda j}}{{s + \lambda
 \left( j + 1 \right)}}\right)\left( {\frac{{\lambda L}}{{s + \lambda L}}}
    \right)^{\mathcal{L} - i}, \\
\prod\limits_{j = i}^{k - 1} \left(\frac{{\lambda j}}{{s + \lambda
 \left( j + 1 \right)}}\right),
\end{array}} \right.
\begin{array}{*{20}{l}}
{{\text{if }}i \le \mathcal{L}},\\
{{\text{if }}\mathcal{L} < i \le \mathcal{U}}.
\end{array}
\]
\item Case 3. $k > \mathcal{U}$, there are three subcases of the expressions depending on $i$:
\[
\prod\limits_{j = i}^{k - 1}
 \left(\frac{{\lambda \widehat{j}}}{{s + \lambda
 \left( \widehat{j + 1} \right)}}\right)
 = 
\left\{ {\begin{array}{*{20}{c}}
\left(\frac{\mathcal{U}}{U}\right)
\left(\frac{{\lambda U}}{{s + \lambda
U }}\right)^{k - \mathcal{U}}
\left(\frac{L}{\mathcal{L}}\right)
\prod\limits_{j = \mathcal{L}}^{\mathcal{U} - 1} \left(\frac{{\lambda j}}{{s + \lambda
 \left( j + 1 \right)}}\right)\left( {\frac{{\lambda L}}{{s + \lambda L}}}
    \right)^{\mathcal{L} - i}, \\ \left(\frac{L}{\mathcal{L}}\right)
\prod\limits_{j = \mathcal{L}}^{k - 1} \left(\frac{{\lambda j}}{{s + \lambda
 \left( j + 1 \right)}}\right)\left( {\frac{{\lambda L}}{{s + \lambda L}}}
    \right)^{\mathcal{L} - i}, \\
\prod\limits_{j = i}^{k - 1} \left(\frac{{\lambda j}}{{s + \lambda
 \left( j + 1 \right)}}\right),
\end{array}} \right.
\begin{array}{*{20}{l}}
{{\text{if }}i \le \mathcal{L},}\\
{{\text{if }}\mathcal{L} < i \le \mathcal{U},}\\
{{\text{if }}i > \mathcal{U}.}
\end{array}
\]
\end{itemize}

To simplify the following analysis, in  (\ref{def:tilde_p}), under the given initial conditions (i.e., $p_1(0) = 1, p_i(0) = 0$ for all other $i$), an approximation of $\widetilde{p_1} \approx 1$ and $\widetilde{p_i} \approx 0$ for all other $i$, i.e., $\forall i = \left\{2, \cdots,  N_\mathcal{T} \right\}$, is used. The approximation error $\lambda \hat{i} P_{i, N_\mathcal{T}} \left(s \right) - e^{-s \mathcal{T}} p_i(\mathcal{T}) \approx 0$ for most values of $i$. It is because $\lambda \hat{i} P_{i, N_\mathcal{T}} \left( s \right)$ is small (and asymptotically converges to 0 as $\mathcal{T} \to \infty$) in typical settings of system parameters, specifically the upper bound $\mathcal{U}$ is typically non-trivial, i.e., $\mathcal{U} \ll N_{\mathcal{T}}$, so $\hat{i}$ is upper bounded, and thus this part of error $\lambda \hat{i} P_{i, N_\mathcal{T}} \left( s \right)$ vanishes as $\mathcal{T} \to \infty$. Besides that, the other part of error $e^{-s \mathcal{T}} p_i\left( \mathcal{T}\right)$ is also small (and asymptotically converges to 0 as $\mathcal{T} \to \infty$).

Hence, equation (\ref{eqn:Laplace-Trichotomy-iteration}) could be simplified to give the $\mathcal{T}$-truncated Laplace transform of $p_k\left(t \right)$ as follows:
\[
\begin{array}{l}
{P_k}\left( s \right) = \left\{ {\begin{array}{*{20}{c}}
{\frac{\widetilde{p_1}}{\lambda L}{{\left( {\frac{{\lambda L}}{{s + \lambda L}}}
\right)}^k}}, &{{\text{if }}1 < k \le \mathcal{L}}&{}\\
{\frac{{\widetilde{p_1}}}{{\lambda \mathcal{L}}}\mathop \prod
\limits_{i = \mathcal{L}}^{k - 1}
\left( {\frac{{\lambda \left( i \right)}}{{s + \lambda
\left( {i + 1} \right)}}} \right){{\left( {\frac{{\lambda L}}
{{s + \lambda L}}} \right)}^\mathcal{L}}}, &
{{\text{if }}\mathcal{L} < k \le \mathcal{U}}&{}\\
{\frac{{\mathcal{U} \widetilde{p_1}}}{{\lambda \mathcal{L} U }}{{\left( {\frac{{\lambda U}}
{{s + \lambda U}}} \right)}^{k - \mathcal{U}}}\mathop \prod
\limits_{i = \mathcal{L}}^{\mathcal{U} - 1}
\left( {\frac{{\lambda \left( i \right)}}{{s + \lambda
\left( {i + 1} \right)}}} \right){{\left( {\frac{{\lambda L}}
{{s + \lambda L}}} \right)}^\mathcal{L}}}, &
{{\text{if }} \mathcal{U} < k \le N_\mathcal{T}}&{}
\end{array}} \right.
\end{array}
\]

As can be seen from Appendix            \ref{AppendixA:residentialTimeByExtOnly},
the residential-time of the node * has an exponential distribution with parameter $\lambda \gamma$ and a normalization constant $\frac{1}{1 - e^{- \lambda \gamma \mathcal{T}}}$. After averaging ${p_k}\left( t
\right)$ with this residential-time distribution, one gets
\begin{eqnarray} \label{eqn:Trichotomy-pmf}
{p_k} &=& {\mathds{E}_T}\left[ {{p_k}\left( t \right)} \right| N(\mathcal{T}) = N_\mathcal{T}] =
 \int_0^\mathcal{T}  {p_k} \left( t \right) \frac{\lambda \gamma
 {e^{ - \lambda \gamma t}}}{1 - e^{- \lambda \gamma \mathcal{T}}} dt \nonumber\\
&=& \frac{\lambda \gamma}{1 - e^{- \lambda \gamma \mathcal{T}}}  P_k \left( \lambda \gamma \right)\nonumber \\
&=& \left\{ {\begin{array}{*{20}{c}}
 c \cdot {\frac{\lambda \gamma}{\lambda L}{{\left( {\frac{{\lambda L}}
 {{\lambda \gamma + \lambda L}}} \right)}^k}}, &{{\text{if }}1 \le k
 \le \mathcal{L}}&{}\\
c \cdot {\frac{{\lambda \gamma}}{{\lambda \mathcal{L}}}\mathop \prod \limits_{i = \mathcal{L}}^{k - 1}
 \left( {\frac{{\lambda \left( i \right)}}{{\lambda \gamma
 + \lambda \left( {i + 1} \right)}}}
 \right){{\left( {\frac{{\lambda L}}{{\lambda \gamma+ \lambda L}}}
 \right)}^\mathcal{L}}}, &{{\text{if }}\mathcal{L} < k \le \mathcal{U}}&{}\\
c \cdot  {\frac{{\lambda \gamma \mathcal{U}}}{{\lambda \mathcal{L} U}}{{\left( {\frac{{\lambda U}}
 {{\lambda \gamma+ \lambda U}}} \right)}^{k - \mathcal{U}}}\mathop
 \prod \limits_{i = \mathcal{L}}^{\mathcal{U} - 1}
 \left( {\frac{{\lambda \left( i \right)}}{{\lambda \gamma+ \lambda
 \left( {i + 1} \right)}}} \right){{\left( {\frac{{\lambda L}}
 {{\lambda \gamma+ \lambda L}}} \right)}^\mathcal{L}}}, &
 {{\text{if }} \mathcal{U} < k \le N_{\mathcal{T}}}
 \end{array}} \right. \nonumber \\
&=& \left\{ {\begin{array}{*{20}{c}} c \cdot 
 {\frac{\gamma}{\gamma + L}{{\left( {\frac{{L}}{{\gamma + L}}}
 \right)}^{k - 1}}}, &{{\text{if }} 1 \le k \le \mathcal{L}}&{}\\
 c \cdot \frac{{\gamma}}{{L}}
 \frac{\frac{{\left( {k - 1} \right)!}}{{\left( {\mathcal{L} - 1}
 \right)!}}}{\frac{\Gamma \left(k + \gamma + 1 \right)}
 {\Gamma \left(\gamma + \mathcal{L} + 1 \right)}}{{\left( {\frac{{L}}
 {{\gamma + L}}} \right)}^\mathcal{L}}, &{{\text{if }}
 \mathcal{L} < k \le \mathcal{U}}\\
c \cdot {\frac{{\gamma}}{{\gamma + U}}{{\left( {\frac{{U}}{{\gamma+ U}}}
 \right)}^{k - \mathcal{U} - 1}} \left( \frac{\mathcal{U}}{\gamma}\right) \left( \frac{\gamma}{L}\right)
 \frac{\frac{{\left( {\mathcal{U} - 1} \right)!}}
 {{\left( {\mathcal{L} - 1} \right)!}}}{\frac{\Gamma
 \left(\mathcal{U} + \gamma + 1 \right)}{\Gamma
 \left(\gamma + \mathcal{L} + 1 \right)}}{{\left( {\frac{{L}}
 {{\gamma + L}}} \right)}^\mathcal{L}}}, &{{\text{if }} \mathcal{U} < k \le N_{\mathcal{T}}}
 \end{array}} \right.,
\end{eqnarray}
where $c$ is the normalization constant such that the sum of the probabilities is 1.

Based on the asymptotic property of the Gamma function, i.e.,
$\lim_{k \to \infty} \frac{\Gamma \left(k \right) k^{\left( \gamma
+ 1 \right)}}{\Gamma \left(k + \gamma + 1 \right)} = 1$, one obtains a
power law with exponent $- \left(\gamma + 1\right)$ for
${\mathcal{L} < k \le \mathcal{U}}$, so one can  obtain the expression for the degree distribution as stated in the theorem.

\end{proof}

It is noted that there is an approximation error term in $p_k$, for very large $k$ (corresponding to the approximation error term in $P_k \left(s \right)$ for very large $k$, i.e., $\lambda \hat{k} P_{k, N_\mathcal{T}} \left(s \right) - e^{-s \mathcal{T}} p_k(\mathcal{T})$), it is $\lambda \hat{k} P_{k, N_\mathcal{T}} \left(\lambda \gamma \right) - e^{-\lambda \gamma \mathcal{T}} p_i(\mathcal{T})$. It will cause some small discrepancies in observed probability density against the theoretical expression in equation (\ref{eqn:Trichotomy-pmf}) for very large $k$, which is usually called "node dynamics" in the existing literature for the power-law case. More details on this node-dynamics will be discussed in Section \ref{sec:sim} below.

\section{Simulations}
\label{sec:sim}

\subsection{Simulation Setting}
\label{subsec:sim-setting}

The node-degree distribution of the MC model has been simulated, for a network of $N_{\mathcal{T}} = 100,000$ nodes. The simulation is set up according to the model as follows: each time when a new node comes, it will connect to an existing node according to the MC model. At the $100,000$-th time unit, there are $100,000$ nodes and the empirical node-degree distribution is reported. The only exception is the simulation for the Poisson model, where there are $50,000$ nodes of degree 0 at the beginning. At the $t$-th time unit, the $(50,000+t)$-th node comes and connects to the existing nodes according to the MC model. At the $50,000$-th time step, there are $100,000$ nodes and the distribution of the first $50,000$ nodes is reported for the Poisson model simulation. Every simulation is repeated for 100 times, and the average result is reported.

\subsection{Major Simulation Results}
\label{subsec:sim-major-results}
First, by setting $L = \mathcal{L} = U = \mathcal{U} = 1$ in the MC model (for accounting node-degree distribution of the  fixed $50,000$ nodes) it reduces to the Poisson model, and the results are shown in Fig. \ref{fig:sim}(a); when excluding isolated nodes, it reduces to the exponential network model, with results shown in Fig. \ref{fig:sim}(b). Second, by setting $L = \mathcal{L} = 1$ and $U = \mathcal{U} = N$, the MC model reduces to the BA model, and the results are shown in Fig. \ref{fig:sim}(c).

Finally, to show a general scenario, by setting $L = \mathcal{L} = 2$ and $U = \mathcal{U} = 8$, the MC model generates the trichotomy
distribution shown in Fig. \ref{fig:sim}(d), i.e. a power-law distribution with exponential head and tail. Several characteristics of the node-degree distribution plot matches with those predicted by theory. The exponent of the power-law region, $-\left(\gamma + 1\right)$, depends on $L$: specifically, $L \leq \gamma \leq L + 1$, particularly for small $U$ as in this case, $\gamma \approx L$. In Fig. \ref{fig:sim}(d), $L = 2$, which matches with the observed exponent $\gamma + 1 = 3$. It is also observed that the head part is geometrically distributed with parameter $0.6$ and the tail part is geometrically distributed with parameter $0.27$.

As a second illustration, the same predictions also hold when $L = \mathcal{L} = 3$ and $U = \mathcal{U} = 10$, as shown in Fig. \ref{fig:sim}(e). Referring to the middle range of \ref{eqn:trichotomyLaw-special}, and $\gamma \approx L$ (for $\mathcal{U} \ll N$), the theory predicts that the slope is $- \left(L + 1\right)$ and it matches with the observed exponent $-4$ in Fig. \ref{fig:sim}(e). It is also observed that the head part is geometrically distributed with parameter value $0.53$ and the tail part is geometrically distributed with parameter value $0.25$, as predicted by the theory.

\subsection{Discussions on node dynamics via Simulation Results}

Each simulation curve presented above is an averaging of the empirical distribution results over $M = 100$ simulation runs, where the goal of these multiple simulation runs is to reduce the variation in the empirical distribution plot, particularly reduce the variation in the tail part, so that the main characteristic of the result can be observed more easily. In the following discussion, the question, why and to what extent a large value of $M$ can reduce the variance, will be addressed. In fact, in one simulation run, suppose $k_i$ are independent for all $i$, according to the law of large numbers, $\frac{1}{N_{\mathcal{T}}} \sum\limits_{i = 1}^{N_{\mathcal{T}}} \mathds{1} \left\{k_i = k\right\} \overset{a.s.}{\to} \mathds{E} \left(\mathds{1} \left\{k_i = k\right\} \right) = p_k$ as $N_{\mathcal{T}} \to \infty$ (where $\mathds{1}\left\{\cdot\right\}$ is an indicator function), i.e., the empirical distribution will converge almost surely (a.s.) to the true degree distribution, which is approximately true in the above setting, i.e. $N_{\mathcal{T}} = 100000$. To avoid possible complications of dependence, the setting on $M$ ($M = 100$) is already large enough to ensure a small variation for most parts of the curve (i.e., for most values of $k$), even when $N_{\mathcal{T}}$ is not large, as $M$ independent simulation runs are performed. However, there is still a significant variation in the tail part (when $k$ is very large) of the empirical distribution observed from the simulation, especially when $U$ is large. For example, $U = \infty$ yields the BA model, but from the proposed MC model, there is significant variation in the tail part as shown in Fig. \ref{fig:sim}(f) even if the simulation runs 100 times. These tail variations in BA model are usually called the "node dynamics" in the existing literature, which have not been theoretically addressed. In the proposed MC model, it corresponds to the neglected boundary term $\lambda \hat{k} p_{k, N_T} (t)$ in equation (\ref{eqn:DE-on-pmf-implicit-w-BdyEffect}), to obtain the degree dynamics in equation (\ref{eqn:DE-on-pmf-implicit}), which causes an approximation error $\lambda \hat{k} P_{k, N_\mathcal{T}} \left(\lambda \gamma \right) - e^{-\lambda \gamma \mathcal{T}} p_k(\mathcal{T})$ as shown in the proof of Theorem \ref{thm:trichotomyLaw}. Hence, the larger this term $\lambda \hat{k} P_{k, N_\mathcal{T}} \left(\lambda \gamma \right)$ is, the more discrepancy it causes to the approximation formula and the empirical curves. As a result, from the proposed MC model viewpoint, the node dynamics under a larger $U$ are more complex, as $\hat{k}$ can increase further with a looser upper bound causing more discrepancy (or dynamics), while the node dynamics for the model with a smaller $U$ is less complex, as $\hat{k}$ cannot increase too far due to the tighter upper bound causing less discrepancy (or dynamics). This prediction is verified in the simulation as shown in Fig. \ref{fig:sim}(e) (where $L = 3, U = 10$) and Fig. \ref{fig:sim}(f) (where $L = 3, U = N$).

\begin{figure}[htbp!]
  \centering
  \subfigure[$L=U=1$, accounting isolated nodes]{
    \includegraphics[width=\figurewidthC]{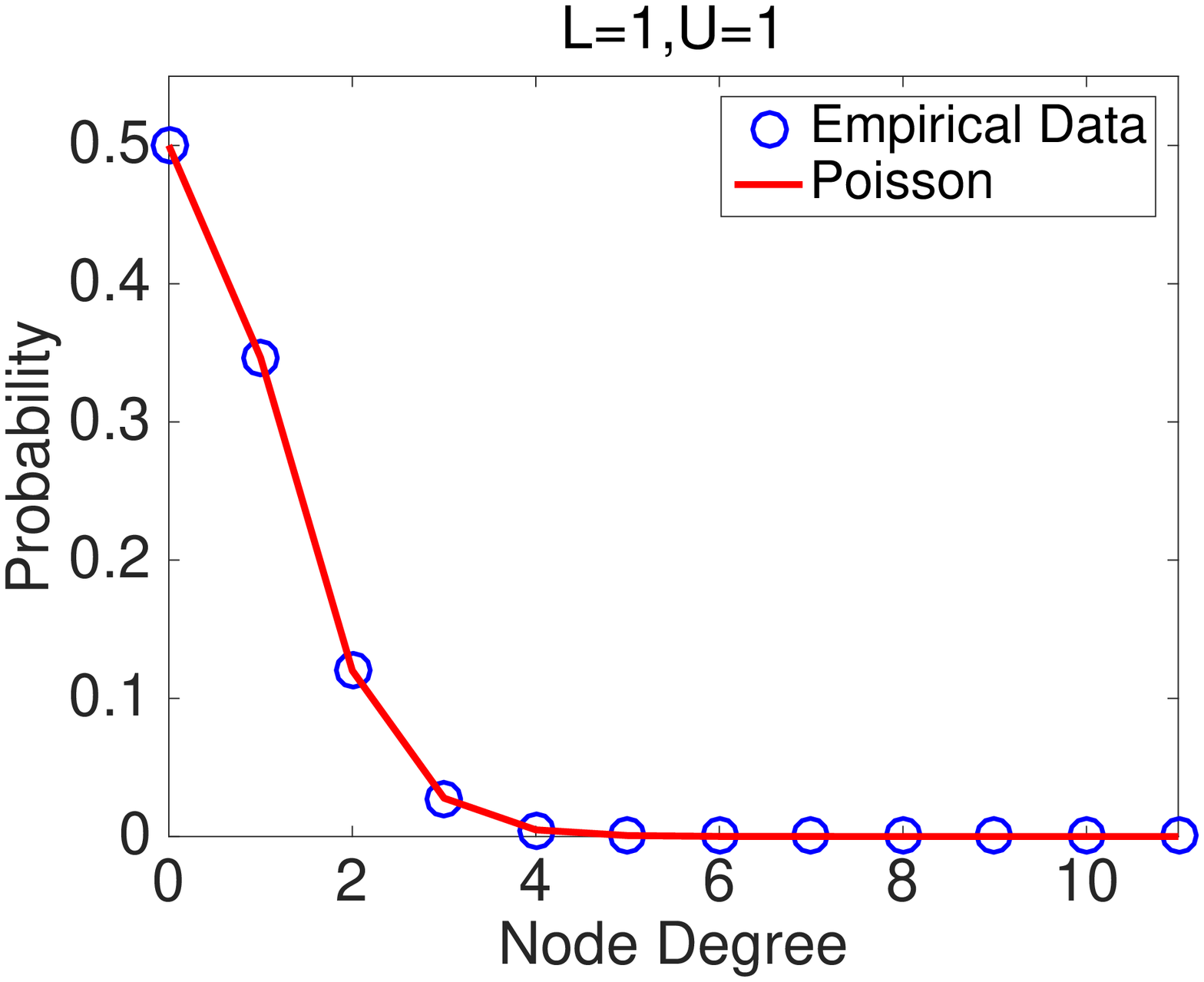}
  }
  \subfigure[$L=U=1$, excluding isolated nodes]{
    \includegraphics[width=\figurewidthC]{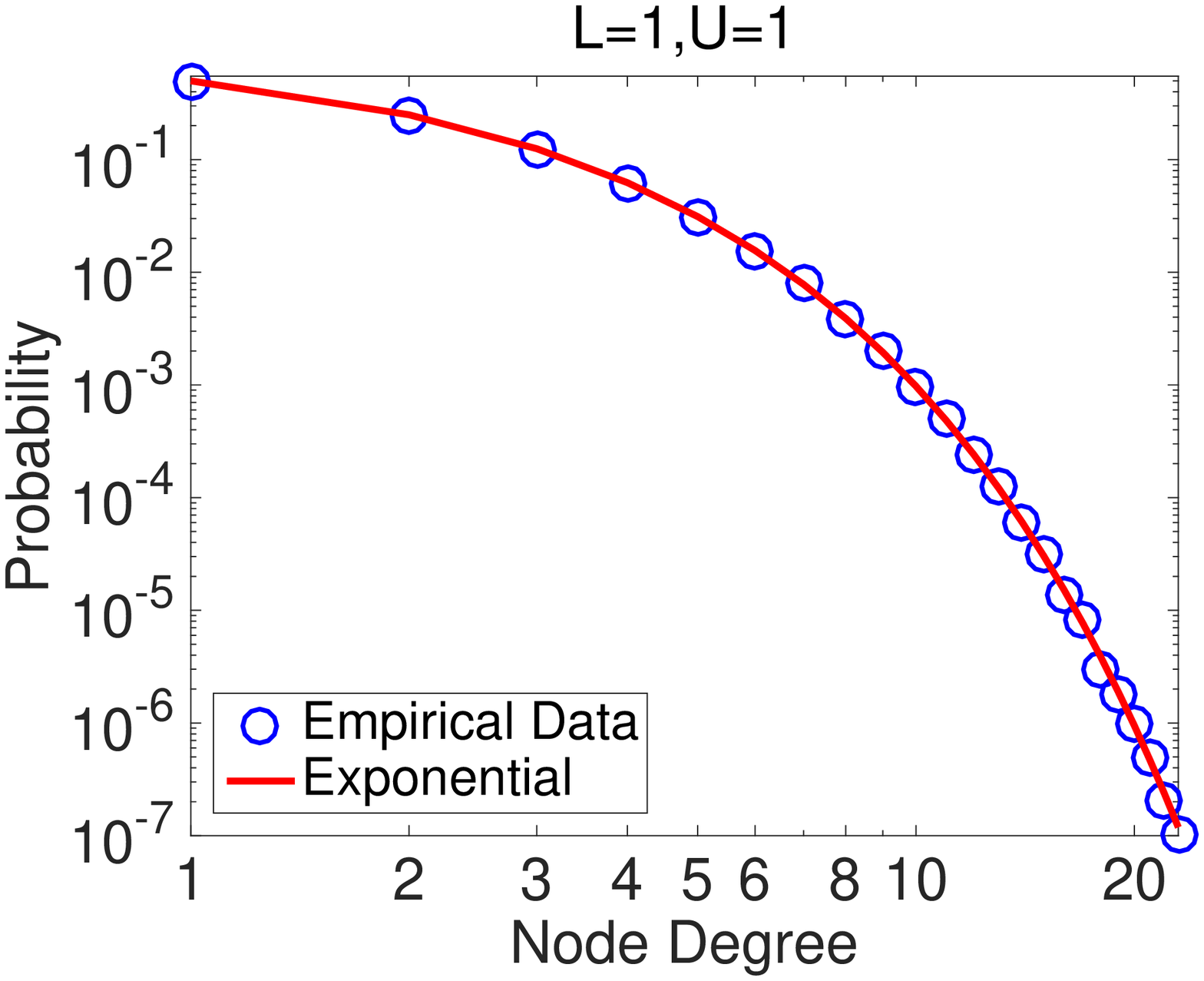}
  }
  \subfigure[$L=1$ and $U=N$]{
    \includegraphics[width=\figurewidthC]{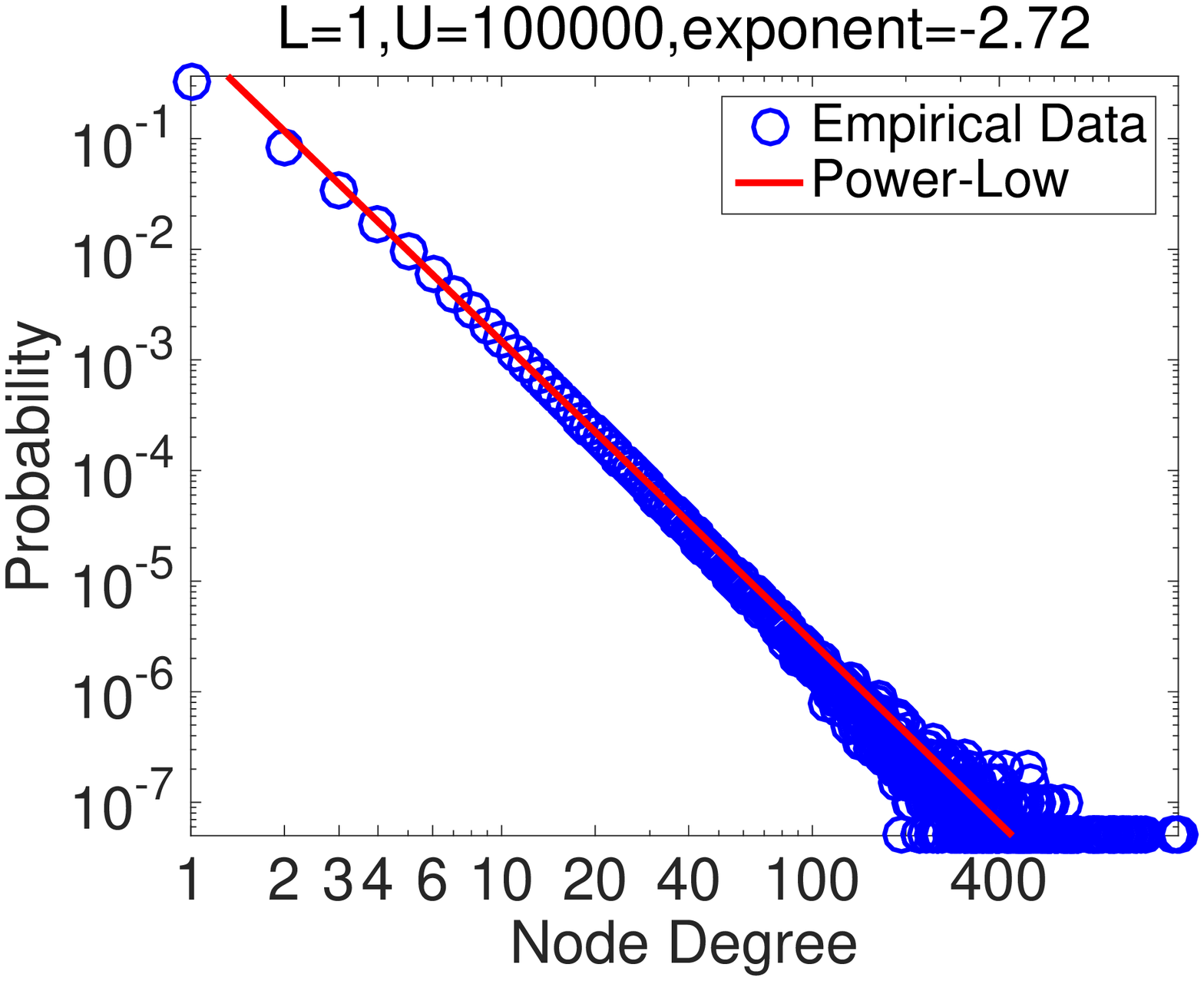}
  }
  \subfigure[$L=2$ and $U=8$]{
    \includegraphics[width=\figurewidthC]{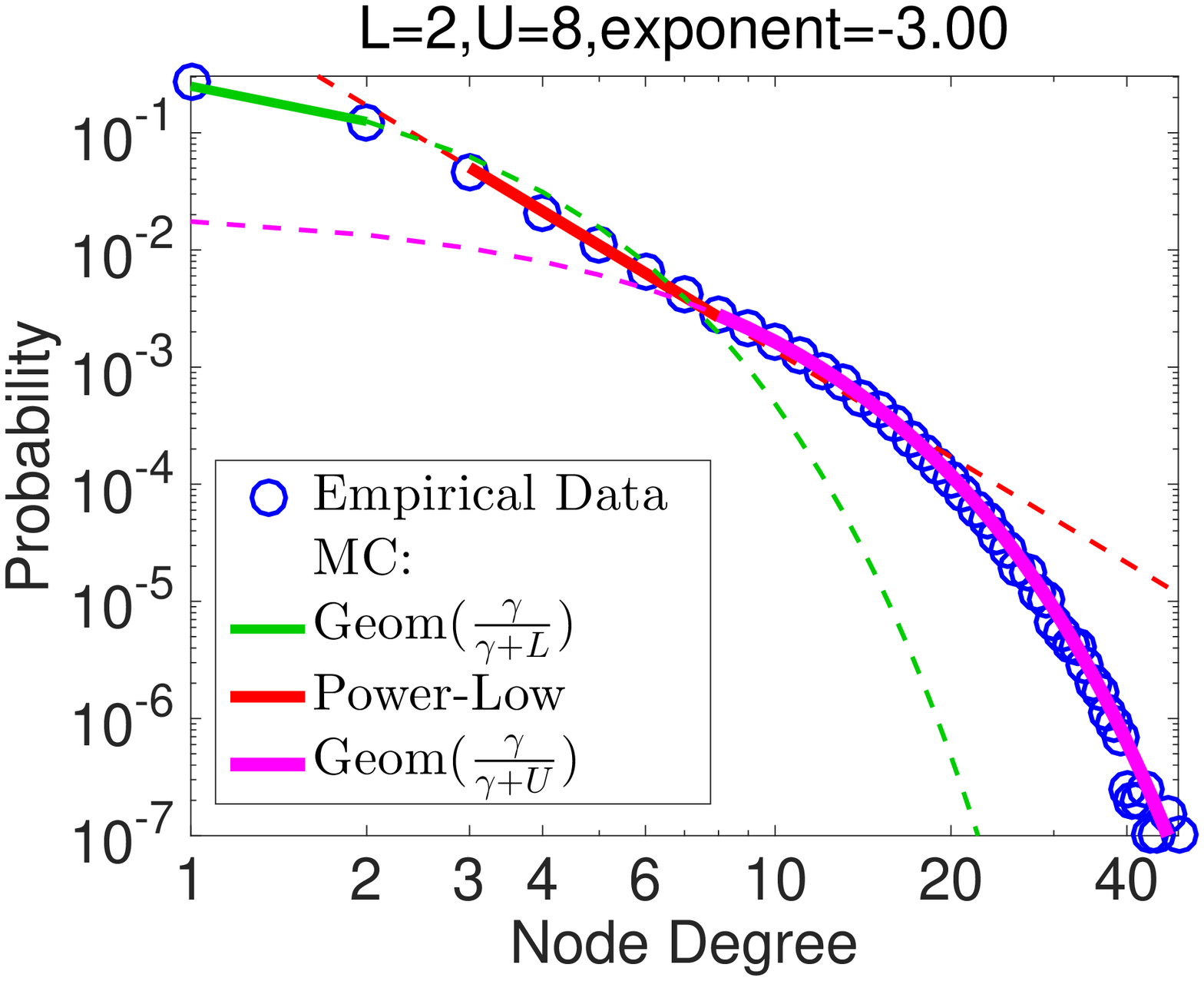}
  }
  \subfigure[$L=3$, $U=10$]{
    \includegraphics[width=\figurewidthC]{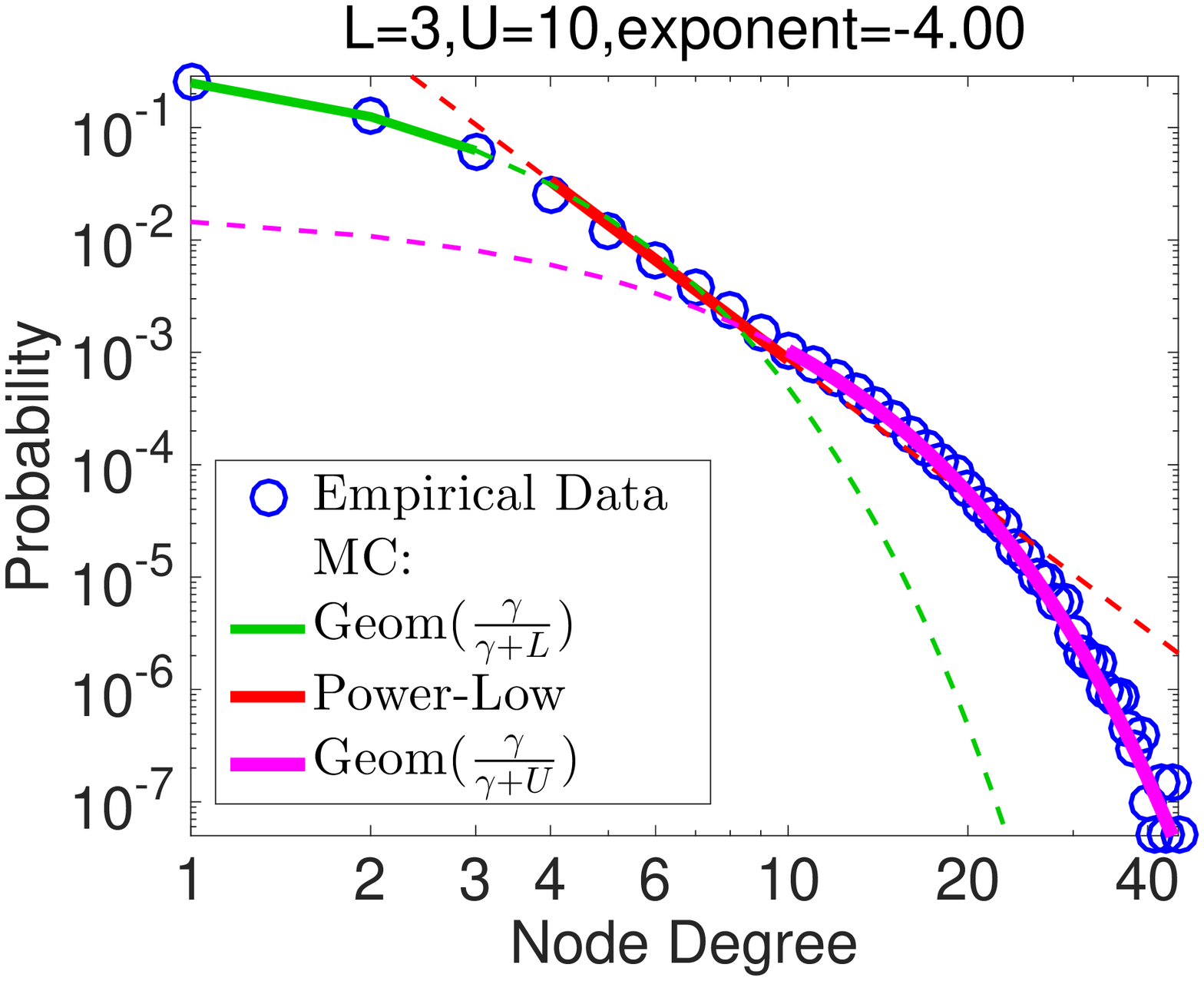}
  }
  \subfigure[$L=3$, $U=N$]{
    \includegraphics[width=\figurewidthC]{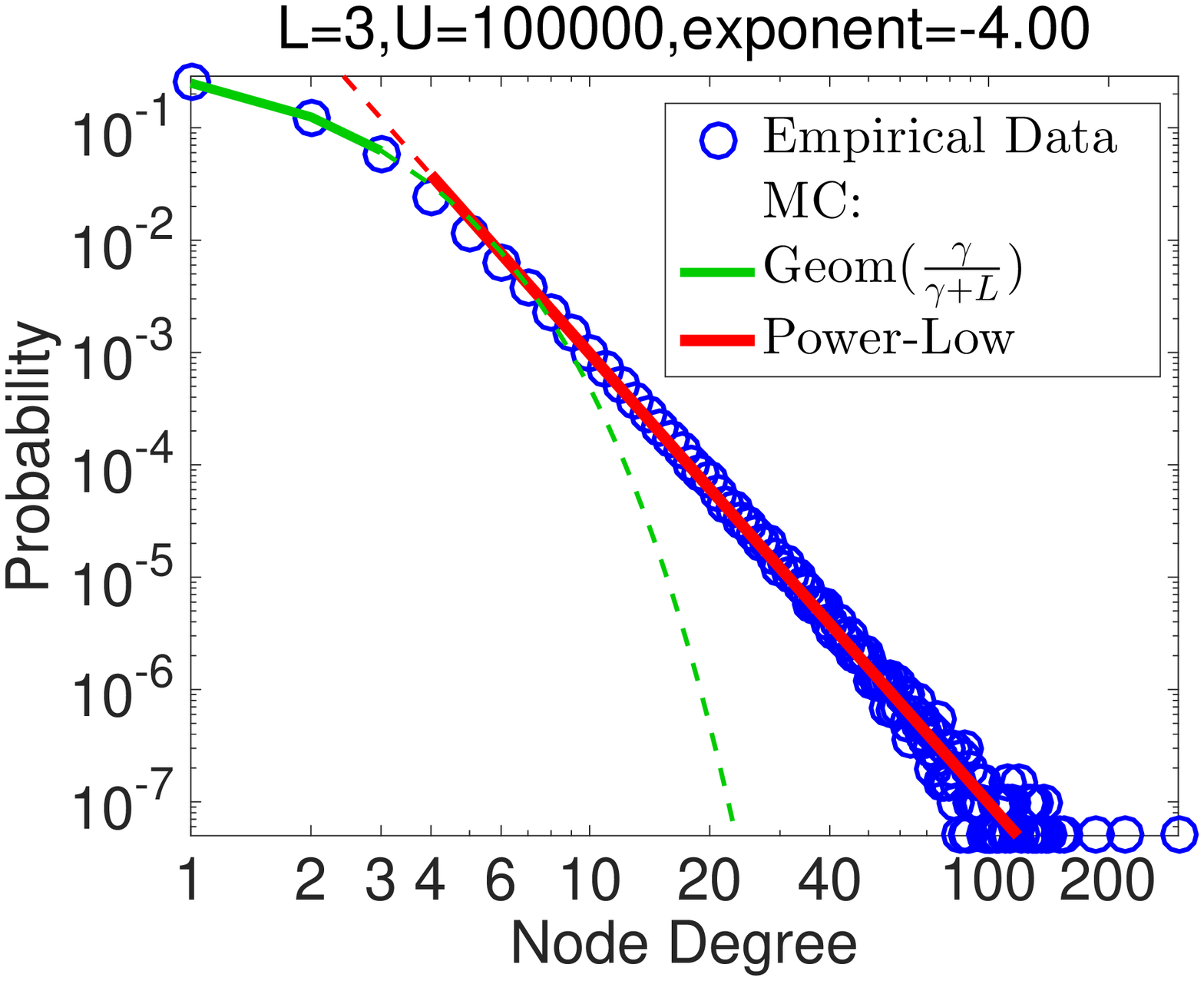}
  }
  \caption{Simulation results of node-degree distribution on the MC model with varying $L$ and $U$.}
  \label{fig:sim}
\end{figure}

\section{Fitting Methodology}
\label{sec:fitting}

A heuristic method is proposed to fit the empirical data. In the following, $N_{\mathcal{T}}$ denotes the number of data points.

\para{Step 1. Fit the fast-evolving phase} (i.e., the phase exhibiting power-law behaviour):
We first fit the fast-evolving phase. Given initial values $\mathcal{L}=L_0$ and $\mathcal{U}=U_0$, we apply least-squares fitting to fit the segment $[\mathcal{L}, \mathcal{U}]$ with a power-law $p_{phase2} = a \times k^{-\gamma}$, where $a$ is a coefficient and $-\gamma$ is the exponent. We use an iterative algorithm where $\mathcal{L}$ and $\mathcal{U}$ are shifted gradually until we find a set of $\mathcal{L}$ and $\mathcal{U}$ leading to the minimal mean square error. To be more specific, we first fix $\mathcal{U}$ while reducing $\mathcal{L}$ by $1$ and calculate the mean square error of the new segment in each iteration. The process ends when the mean square error is not reduced anymore or $\mathcal{L}=1$. Similarly, we fix $\mathcal{U}$ while increasing $\mathcal{L}$ by $1$ and calculating mean square errors iteratively. By comparing the mean square errors while reducing and increasing $\mathcal{L}$, we can find the best $\mathcal{L}$ with the minimal mean square error. We apply the same process to reduce and increase $\mathcal{U}$ while fixing $\mathcal{L}$, and find the best $\mathcal{U}$ with the minimal mean square error. Note that, although this is a heuristic method which could be vulnerable to noise and only find a local optimal set of $\mathcal{L}$ and $\mathcal{U}$, our evaluation in Section \ref{sec:exp} shows that the method works reasonably good for various types of real-world datasets.

\para{Step 2. Fitting the initializing phase} (i.e., the phase exhibiting \textit{geom}($\frac{\gamma}{L+\gamma}$) distribution):
With the exponent found in Step 1, we use the derived closed-form formula to fit the initializing phase. Since it is usually unclear how many truncated geometric distributions are combined in this phase, we use only two truncated geometric distributions (i.e., the maximum number of head parameter is set to 1). According to the equation, $p_{phase1} = p_1^0 \times p_a \times (1-p_a)^{k-1} \mathds{1}_{\left\{k \geq 1\right\}} + \left(1 - p_1^0\right) \times p_{a} \times (1-p_{a})^{k-2} \mathds{1}_{\left\{k \geq 2\right\}}$, where $p_{a} = \gamma/(L+\gamma)$. We then apply least-squares fitting to find the best $p_1^0$ so as to fit the initializing phase segment $[1, \mathcal{L}]$. For the general case of fitting the initializing phase (with a larger maximum number of head parameters), it is referred to Fig. \ref{fig:fitting-headmethod}.

\para{Step 3. Fitting the maturing phase} (i.e., the phase exhibiting the \textit{geom}($\frac{\gamma}{U+\gamma}$) distribution):
Finally, we fit the maturing phase by the derived closed-form formula $p_{phase3} = c \times p_{b} \times (1-p_{b})^{k-1}$, where $c$ is the coefficient and $p_{b} = \gamma/(U+\gamma)$. We again use least-squares fitting to find the best $c$ and fit the maturing phase segment $[\mathcal{U}, N_{\mathcal{T}}]$.

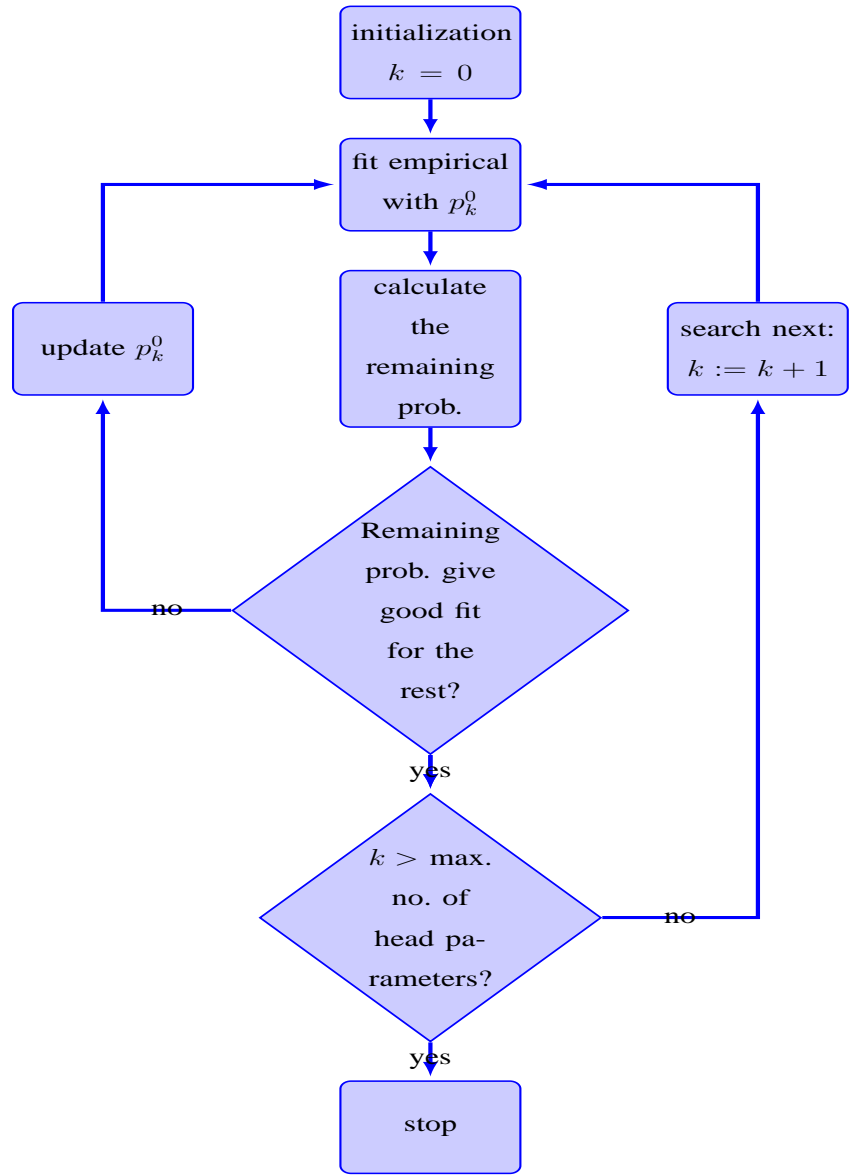
\begin{figure}[htbp!]
	\centering
\resizebox{0.7\textwidth}{0.7\textheight}{%
\begin{tikzpicture}
[
	%
	DecisionStyle/.style =
	{
		diamond,
		draw		= blue,
		thick,
		fill		= blue!20,
		text width	= 4.5em,
		align		= flush center,
		inner sep	= 1pt
	},
	BlockStyle/.style =
	{
		rectangle,
		draw			= blue,
		thick,
		fill			= blue!20,
		text width		= 5em,
		align			= center,
		rounded corners,
		minimum height	= 4em
	},
	CloudLineStyle/.style =
	{
		draw,
		ultra thick,
		color	= red,
		-latex,
		shorten	>= 2pt,
		dotted
	},
	BlockLineStyle/.style =
	{
		draw,
		ultra thick,
		color	= blue,
		-latex,
		shorten	>= 2pt
	},
	CloudStyle/.style =
	{
		draw			= red,
		thick,
		ellipse,
		fill			= red!20,
		minimum height	= 2em
	}
]

\matrix [column sep = 5mm, row sep = 7mm, ampersand replacement=\&]
{
                                                                               \&
\node [BlockStyle] (init)	{initialization $k = 0$};	\&
                                                                               \\
%
															\&
\node [BlockStyle] (identify)	{fit empirical with $p_k^0$};	\&
															\\
%
\node [BlockStyle] (update)		{update $p_k^0$};					\&
\node [BlockStyle] (evaluate)	{calculate the remaining prob.};	\&
\node [BlockStyle] (search_next)   	{search next: \\ $k := k +1$};		\\
%
																\&
\node [DecisionStyle] (decideGoodFit)	{Remaining prob. give good fit for the rest?};		
               \&
																\\
												
%
																\&
\node [DecisionStyle] (decideFinish)	{$k >$  max. no. of head parameters?};		
               \&
               				\\
																\&
\node [BlockStyle] (stop)	{stop};								\&
																\\
}; 

\draw	[BlockLineStyle]	(init)		-- (identify);
\draw	[BlockLineStyle]	(identify)	-- (evaluate);
\draw	[BlockLineStyle]	(evaluate)	-- (decideGoodFit);
\draw	[BlockLineStyle]	(decideGoodFit)	-- (decideFinish);
\draw	[BlockLineStyle]	(update)	|- (identify);
\draw	[BlockLineStyle]	(search_next) |- (identify);
\draw	[BlockLineStyle]	(decideGoodFit)	-| (update)		node [near start, color = black]	{no} ;
\draw	[BlockLineStyle]	(decideGoodFit)	-- (decideFinish)		node [midway, color = black] {yes};
\draw	[BlockLineStyle]	(decideFinish)	-| (search_next)		node [near start, color = black]	{no} ;
\draw	[BlockLineStyle]	(decideFinish)	-- (stop)		node [midway, color = black]		{yes} ;

\end{tikzpicture}
}
 		\caption{The method of fitting the head part of the distribution, i.e. the initializing phase}
 		\label{fig:fitting-headmethod}
\end{figure}


\section{Experimental Data}
\label{sec:exp}

We use nine datasets from citation networks, social networks and vehicular networks to verify our model. 

\subsection{Citation Networks}
\begin{table}[htb]
  \centering
  {\small
    \begin{tabular}{|c|c|c|c|}
    \hline
      Dataset & Date & \#Nodes \\\hline
      DBLP Citation \cite{tang:kdd08:citation} & 1995-2014 & $2,146,341$ papers \\\hline
      APS Citation \cite{aps-cit} & 1893-2013 & $531,478$ papers \\\hline
      US Patent Citation \cite{leskovec:kdd05:us-patent-cit} & 1975-1999 & $3,774,768$ patents \\\hline
    \end{tabular}
  }
  \vspace{1pt}
  \caption{Citation Datasets.}
  \label{tab:citation-data}
  \vspace{-30pt}
\end{table}

Table~\ref{tab:citation-data} shows three datasets from scientific publication citations and patent citations \cite{tang:kdd08:citation,aps-cit,leskovec:kdd05:us-patent-cit}, which are used in our analysis. We investigate the distribution of the number of citations for each dataset. 

\para{Physical Meaning: }
In citation networks, papers are presented as nodes. When a paper $p1$ cites another paper $p2$, a link is built between $p1$ and $p2$ and the node degree of each paper increases by $1$. When there is only a few papers in a field, a new paper is likely to cite any of them at random. When the publication number grows, some of them become \textit{famous} for having major findings or presenting the state-of-the-art results (\ie, papers with citations $> \mathcal{L}$). New papers may cite these famous papers with a higher probability. For those papers that are highly reputable (\ie, papers with citations $> \mathcal{U}$), a new paper may cite any of them with the same probability because they are all important. A larger $\mathcal{L}$ in a research area means that a paper needs a larger number of citations before starting to draw significant attention, and a larger $\mathcal{U}$ means that the study in the area may have been very popular. The interval between $L$ and $U$ can represent the amount of work needed to be done to explore, verify, or extend a work so as to make it become one of the most reputed papers in its field.

\begin{figure}[htbp!]
  \centering
  \vspace{-20pt}
  \subfigure[DBLP]{
    \includegraphics[width=\figurewidthD]{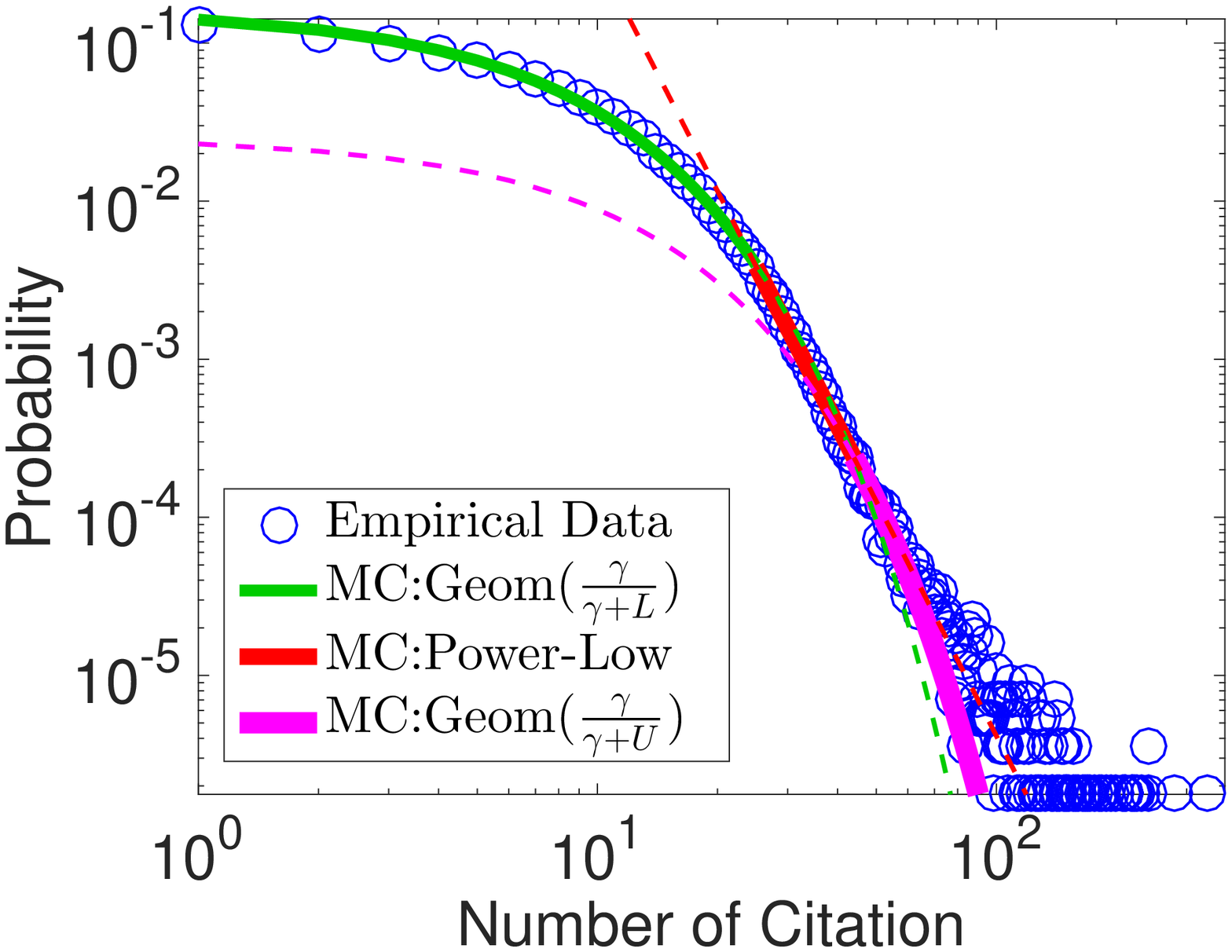}
  }
  \subfigure[DBLP: Networking]{
    \includegraphics[width=\figurewidthD]{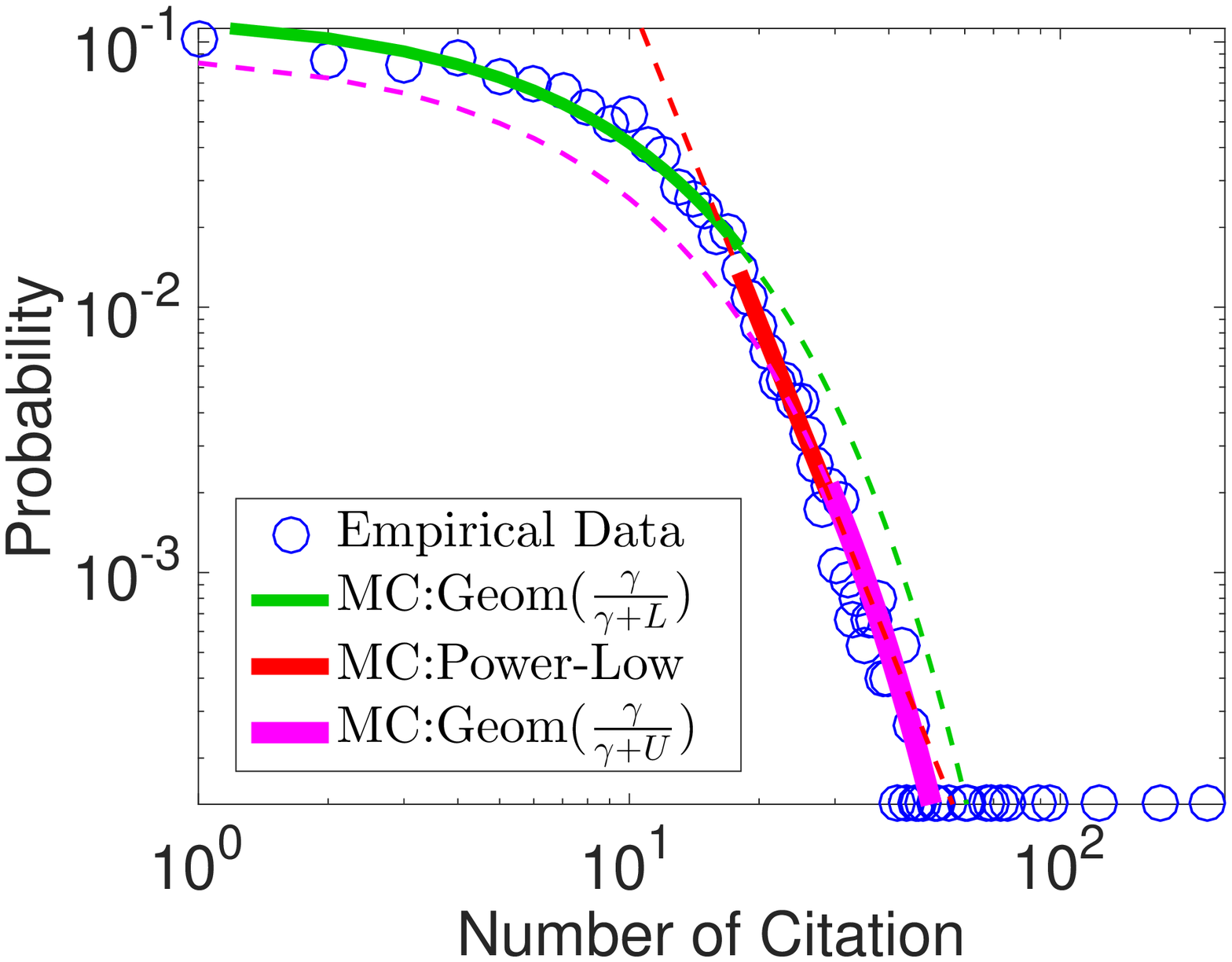}
  }
  \subfigure[APS]{
    \includegraphics[width=\figurewidthD]{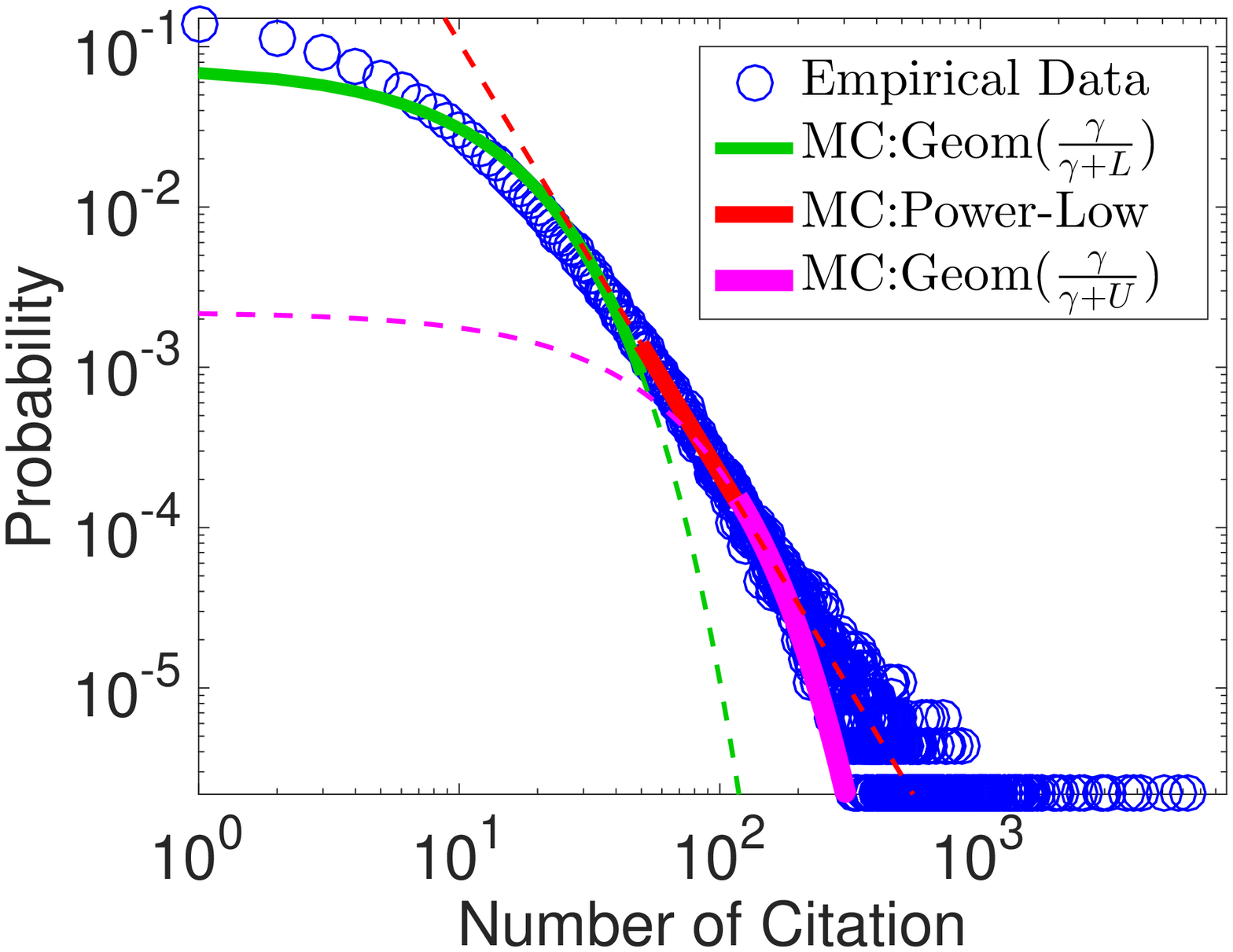}
  }
  \subfigure[US Patent]{
    \includegraphics[width=\figurewidthD]{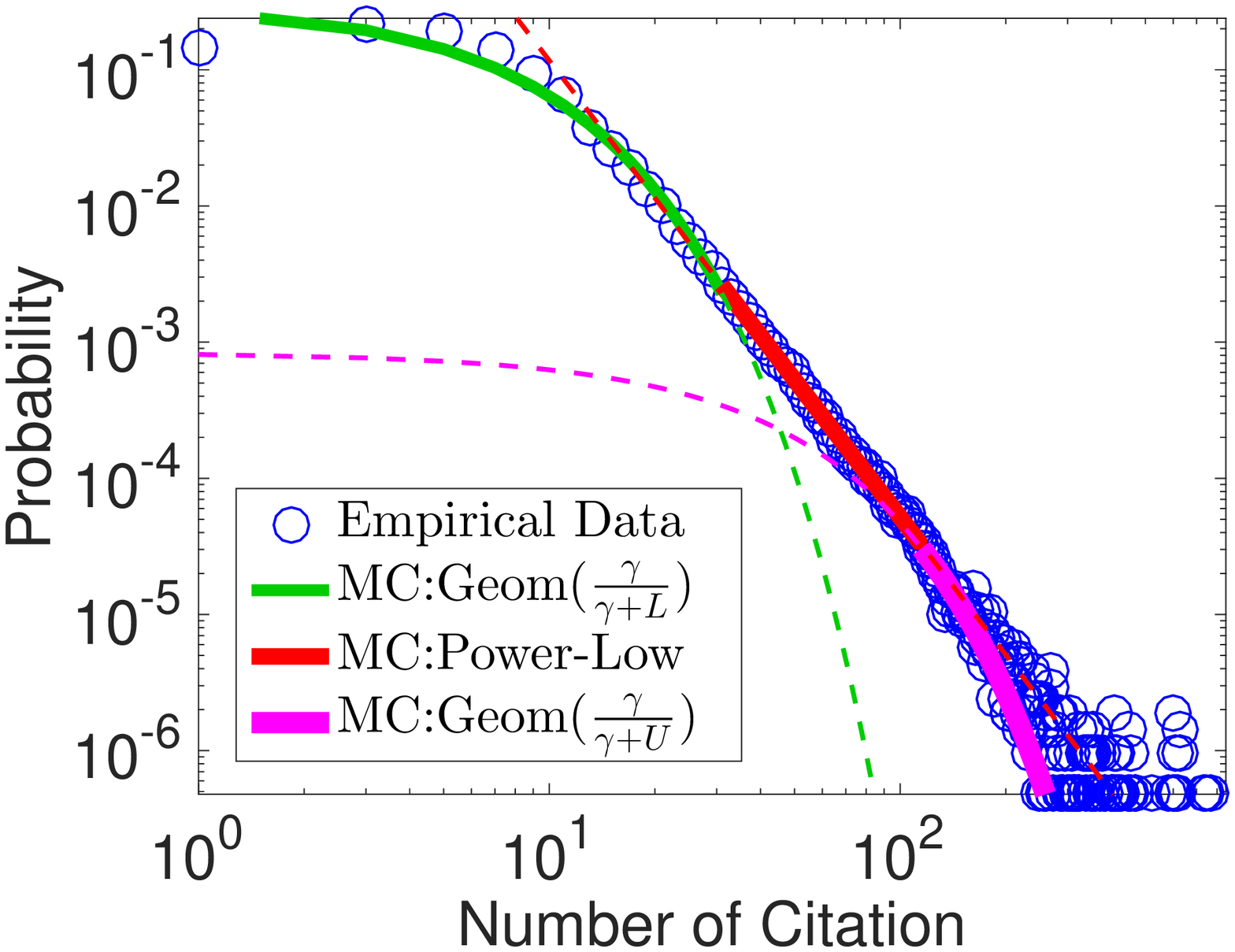}
  }
  \caption{Citation probability $P(x)$ versus the number of citations $x$ on a double logarithmic scale.}
  \label{fig:citation}
\end{figure}

\begin{table}[htb]
  \centering
  {\small
    \begin{tabular}{|c|c|c|c|c|c|}
    \hline
      Dataset & L & U & exponent & RMSE: ours & RMSE: pl \\\hline
      DBLP & $25$ & $44$ & $-4.93$ & $1.41 \times 10^{-4}$ & $92.28$ \\\hline
      Networking & $18$ & $29$ & $-4.05$ & $4.72 \times 10^{-4}$ & 8.05 \\\hline
      APS & $50$ & $117$ & $-2.69$ & $4.3 \times 10^{-5}$ & $0.01$ \\\hline
      US Patent & $31$ & $115$ & $-3.35$ & $6.94 \times 10^{-4}$ & $0.70$ \\\hline
    \end{tabular}
  }
  \vspace{1pt}
  \caption{Fitting parameters and errors in citation datasets. $RMSE: ours$ is the fitting error of our model and $RMSE: pl$ is the fitting error of using Power-Law only.}
  \label{tab:citation-param}
  \vspace{-30pt}
\end{table}

\para{Empirical Data and Fitting Results: }
Fig.~\ref{fig:citation} shows the distributions of citations for DBLP Computer Science publications, American Physical Society (APS) publications, and US patents datasets. Since preferential attachment can be best possible to model citations within one field of research \cite{medo:prl:powerlaw_decay}, we conduct our analysis in a subset of papers about Computer Networks in DBLP datasets as shown in Fig.~\ref{fig:citation}(b). We observe similar behavior, where curves are separated into three segments corresponding to the three phases in the MC model, respectively. 

We compare our MC model with the one which  uses only power-law to fit the data. The fitting error is quantized using root-mean-square error (RMSE = $\sqrt{\sum_{t=1}^{n}(\hat{y}-y)^2/n}$). The fitting parameters and errors are shown in Table~\ref{tab:citation-param}. One can see that our MC model reduces the fitting errors by more than 99\% in all datasets.

\para{Utility: }
The citation number is a common criterion used to evaluate one aspect of a scientist. A larger citation number of a scientist means that he or she has published more attractive works. However, using citation number to compare across different research areas is often misleading. For example, the citation number of an important paper in a young research area may be considered small in a well-studied area. Even in the same field, twice citation number doesn't mean twice importance because the distribution of citation numbers is not linear. Here, we use the MC model to evaluate the citation networks. Given a citation number $k$, by comparing $k$ with $\mathcal{L}$ and $\mathcal{U}$ in the area, one can estimate if the scientist is junior ($k<\mathcal{L}$), senior ($\mathcal{L}<k<\mathcal{U}$), or an expert ($\mathcal{U}<k$). To compare across areas, one can also derive the frequency of the work (i.e., the number of papers with the same level of importance over the total number of papers) in the area by citation numbers. Since the proposed method accounts for different distributions in different areas, the comparison is fairer than simply using citation numbers alone.


\subsection{Social Networks}

\begin{table}[htb]
  \centering
  {\small
    \begin{tabular}{|c|c|c|c|}
    \hline
      Dataset & Date & \# of Nodes \\\hline
      DBLP Coauthor \cite{tang:kdd08:citation} & 1995-2014 & $1,234,706$ authors \\\hline
      Facebook Friendship \cite{snapnets,mcauley:tkdd14:social} & 2012 & $4,039$ users \\\hline
      Twitter Friendship \cite{snapnets,leskovec:nips12:tiwtter} & 2012 & $81,306$ users \\\hline
    \end{tabular}
  }
  \vspace{1pt}
  \caption{Social Networks Datasets.}
  \label{tab:social-data}
\end{table}

Table~\ref{tab:social-data} shows three social network datasets \cite{tang:kdd08:citation,mcauley:tkdd14:social,leskovec:nips12:tiwtter}, which are used in our verification and analysis.  

\para{Physical Meaning: }
In social networks, people are considered as nodes. When a person $p1$ coauthors a paper with another person $p2$, adds $p2$ as a friend on Facebook, or follows $p2$ on Twitter, a link is built between $p1$ and $p2$ and the node degrees of both persons increase by $1$. To further explain the lower bound $L$ and the upper bound $U$ in the new MC model, take the coauthorship network as an example. When there is no \textit{famous scientist} (i.e., a node with high degree), each scientist likely chooses to work with others with an equal probability. When there are more and more scientists in the network, a new scientist has a higher probability to choose to work with those having better reputation (i.e., nodes with degrees $>L$). For those who are already very famous (i.e., nodes with degrees $>U$), a new scientist may choose to work with any of them with an equal probability because they all have a \rm{good enough} reputation. It is a similar process for friendships in Facebook and Twitter.

\para{Empirical Data and Fitting Results: }

\begin{figure}[htbp!]
  \centering
  \subfigure[DBLP]{
    \includegraphics[width=\figurewidthD]{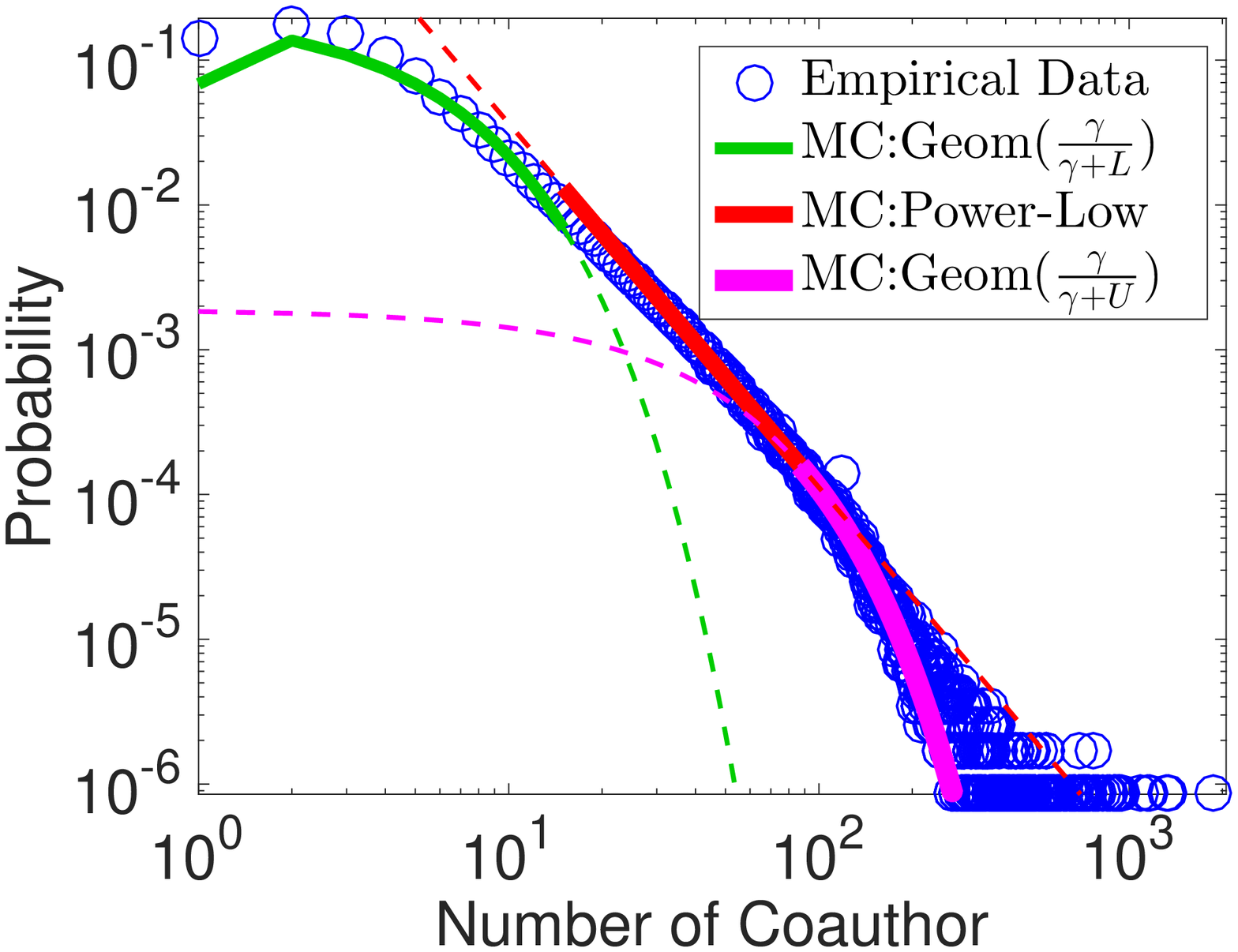}
  }
  \subfigure[DBLP: Networking]{
    \includegraphics[width=\figurewidthD]{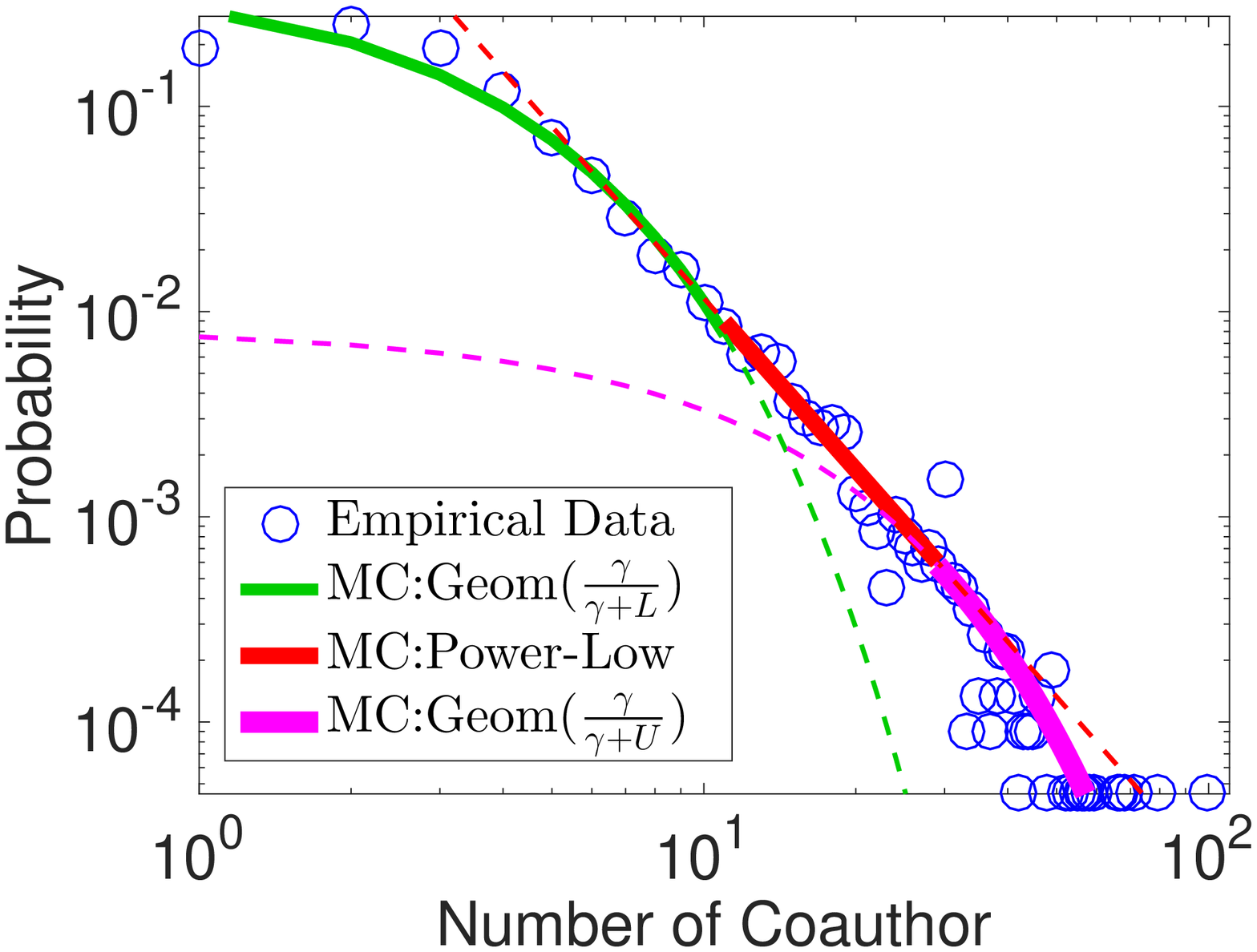}
  }
  \subfigure[Facebook]{
    \includegraphics[width=\figurewidthD]{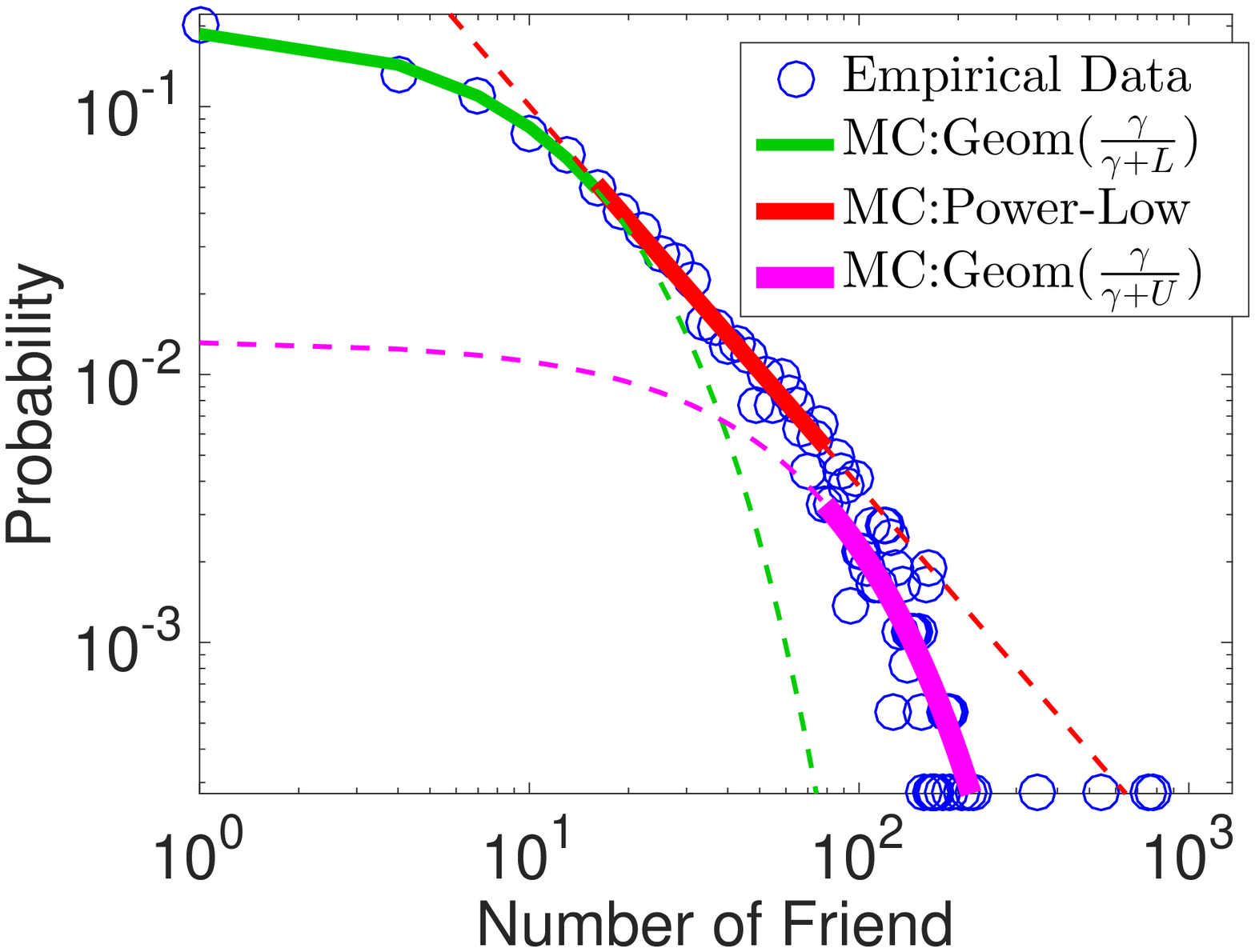}
  }
  \subfigure[Twitter]{
    \includegraphics[width=\figurewidthD]{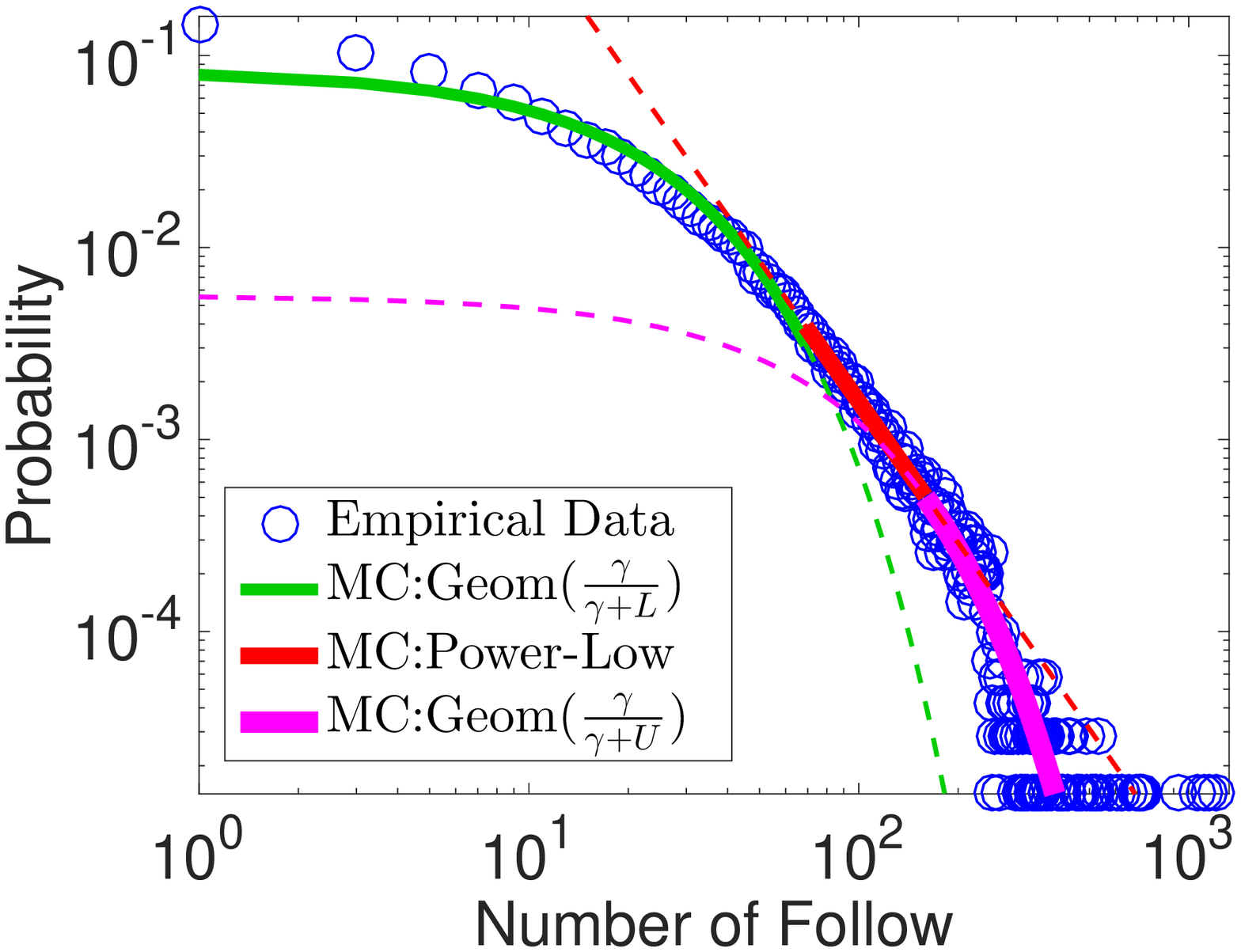}
  }
  \caption{Probability $P(x)$ versus node degree $x$ on a double logarithmic scale in social network datasets. The node degree represents the number of coauthors, friends, or followers for the three datasets, respectively.}
  \label{fig:social}
\end{figure}

\begin{table}[htb]
  \centering
  {\small
    \begin{tabular}{|c|c|c|c|c|c|}
    \hline
      Dataset & $\mathcal{L}$ & $\mathcal{U}$ & exponent & RMSE: ours & RMSE: pl \\\hline
      DBLP & $15$ & $87$ & $-2.51$ & $1.17 \times 10^{-4}$ & $0.01$  \\\hline
      Networking & $11$ & $29$ & $-2.78$ & $2.34 \times 10^{-3}$ & $0.08$  \\\hline
      Facebook & $16$ & $79$ & $-1.42$ & $2.52 \times 10^{-4}$ & $0.01$ \\\hline
      Twitter & $69$ & $159$ & $-2.44$ & $3.18 \times 10^{-4}$ & $0.22$ \\\hline
    \end{tabular}
  }
  \vspace{1pt}
  \caption{Fitting parameters and errors in social networks datasets. $RMSE: ours$ is the fitting error of our model and $RMSE: pl$ is the fitting error by using power-law only.}
  \label{tab:social-param}
\end{table}

Fig.~\ref{fig:social} shows the distributions of coauthors for DBLP Computer Science publications, friends in Facebook datasets, and followers in Twitter datasets. Fig.~\ref{fig:social}(b) further shows the coauthor distribution in a subset of publications about Computer Networks in DBLP dataset. The fitting parameters and errors are shown in Table~\ref{tab:social-param}. One can see that the new MC model reduces the fitting errors by 97\% - 99\%.

\para{Utility: }
In social networks, the distribution of the numbers of friends or coauthors that each people has is an important indicator of the health of the networks. Typically, the more evenly distributed the network is (i.e., the more number of friends each people has), the healthier (i.e., more desirable) it is. In the MC model, this corresponds to a smaller $L$, as the exponent in the power-law region $\gamma$ is proportional to $L$, hence a smaller $L$ implies a flatter slope, which means a flatter distribution. However, as proven in the MC model, a smaller $L$ also means a smaller growth rate of the network. Hence, there is a tradeoff in obtaining a healthier distribution and a faster growth of the network by changing $L$. In practice, an administrator of a social network can control this $L$ in various ways. For example, the administrator can introduce subsidies or promotions to lure the recruitment of new members, hence increasing $L$, or introduce fees for joining the network, hence reducing $L$, so as to obtain a desired tradeoff level.


\subsection{Vehicular Networks}

\begin{table}[htb]
  \centering
  {\small
    \begin{tabular}{|c|c|c|c|}
    \hline
      Dataset & Date & Duration & \# of vehicles  \\\hline
      Rome Taxi~\cite{rome-taxi} & Feb 01, 2014 & 30 days & 316 \\\hline
      Beijing Taxi~\cite{yuan:gis10:beijing-taxi,yuan:kdd11:beijing-taxi} & Feb. 02, 2008 & 7 days &  10,336 \\\hline
      San Francisco Taxi~\cite{sf-taxi} & May 17, 2008 & 24 days & 536 \\\hline
    \end{tabular}
  }
  \vspace{1pt}
  \caption{Vehicular Datasets.}
  \label{tab:v2v-data}
\end{table}

Table~\ref{tab:v2v-data} shows three vehicular networks datasets \cite{rome-taxi,sf-taxi,yuan:gis10:beijing-taxi,yuan:kdd11:beijing-taxi}, which are used in our verification and analysis.  

\para{Physical Meaning: }
In vehicular networks, vehicles are considered as nodes. When a vehicle $v1$ is within the communication range $r$ with another vehicle $v2$, a link is built between $v1$ and $v2$ and the contact counts (i.e., node degrees) of both vehicles increase by $1$. A vehicle's degree is related to the area it locates. For those vehicles located in a quiet zone, there are only a few vehicles around, so their node degrees are low (i.e., degrees $< L$). These vehicles move freely and have an equal probability to meet each other. On the other hand, a vehicle will have a higher degree if it is closer to the crowded areas (e.g., downtowns or hot scenic points). When those vehicles get closer to such areas, they will meet more vehicles with higher probabilities. Vehicles located in such areas may be able to communicate with most vehicles in the areas and therefore can link to these vehicles with an equal probability.

\para{Empirical Data and Fitting Results: }

\begin{figure}[htbp!]
  \centering
  \subfigure[Rome]{
    \includegraphics[width=\figurewidthD]{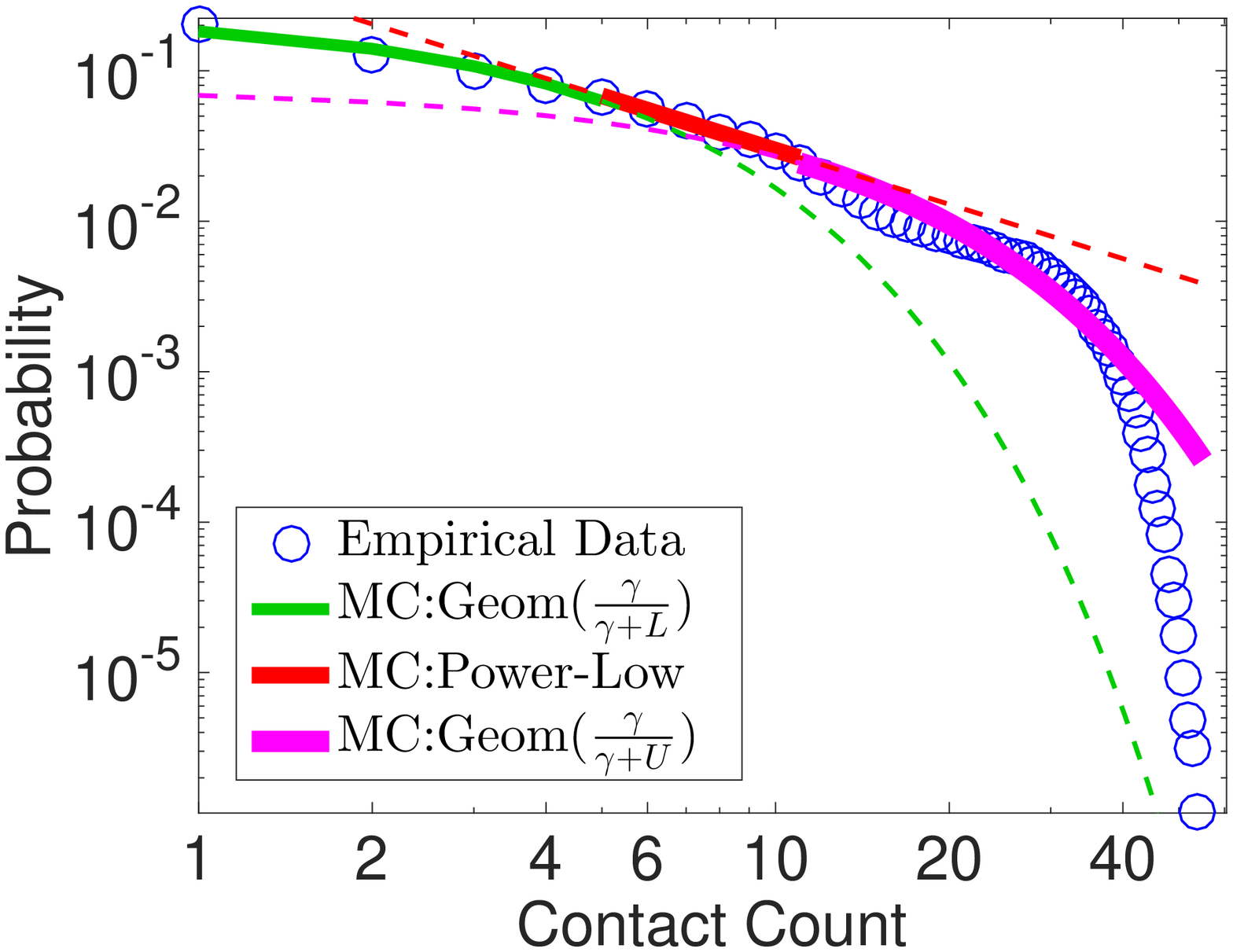}
  }
  \subfigure[Beijing]{
    \includegraphics[width=\figurewidthD]{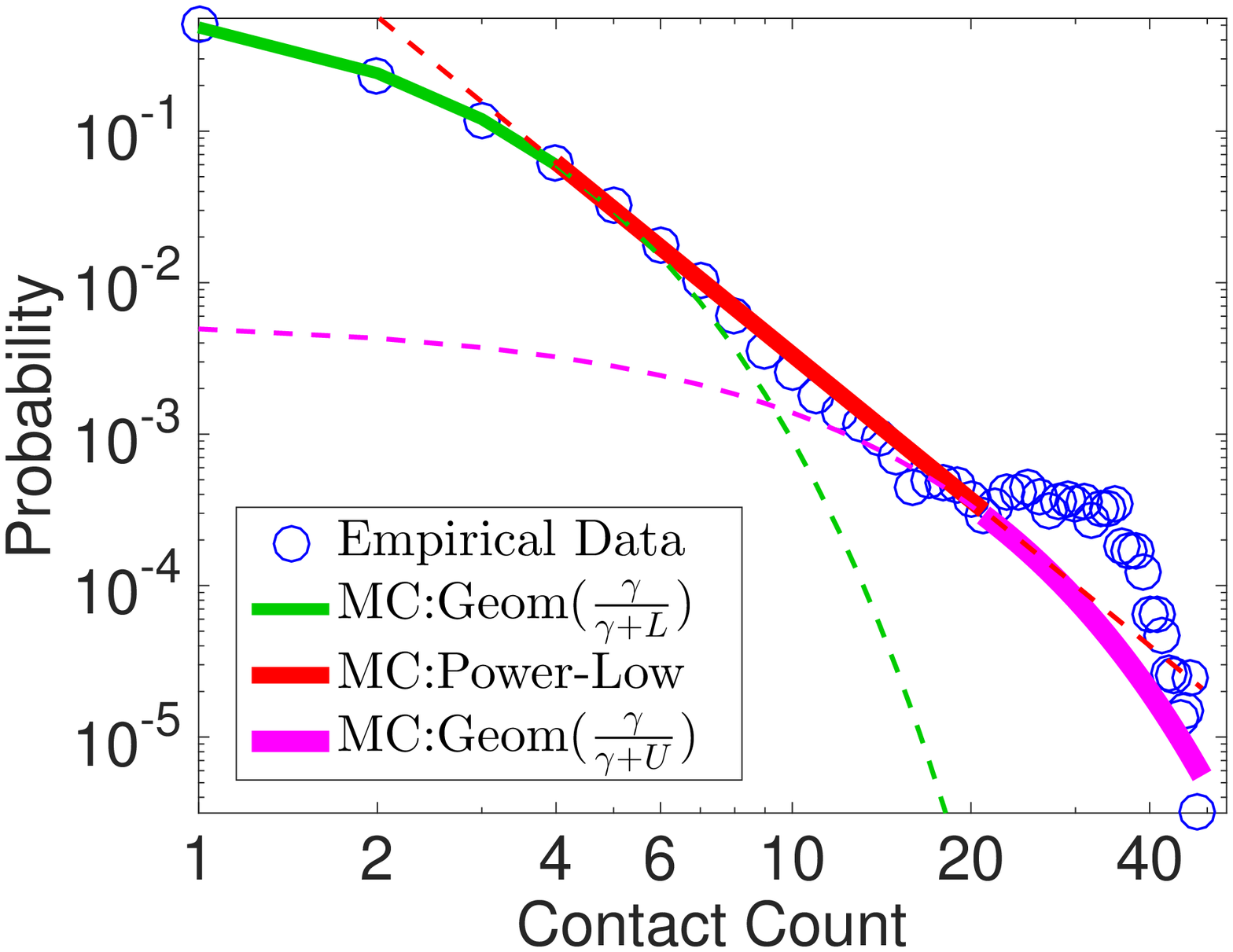}
  }
  \subfigure[San Francisco]{
    \includegraphics[width=\figurewidthD]{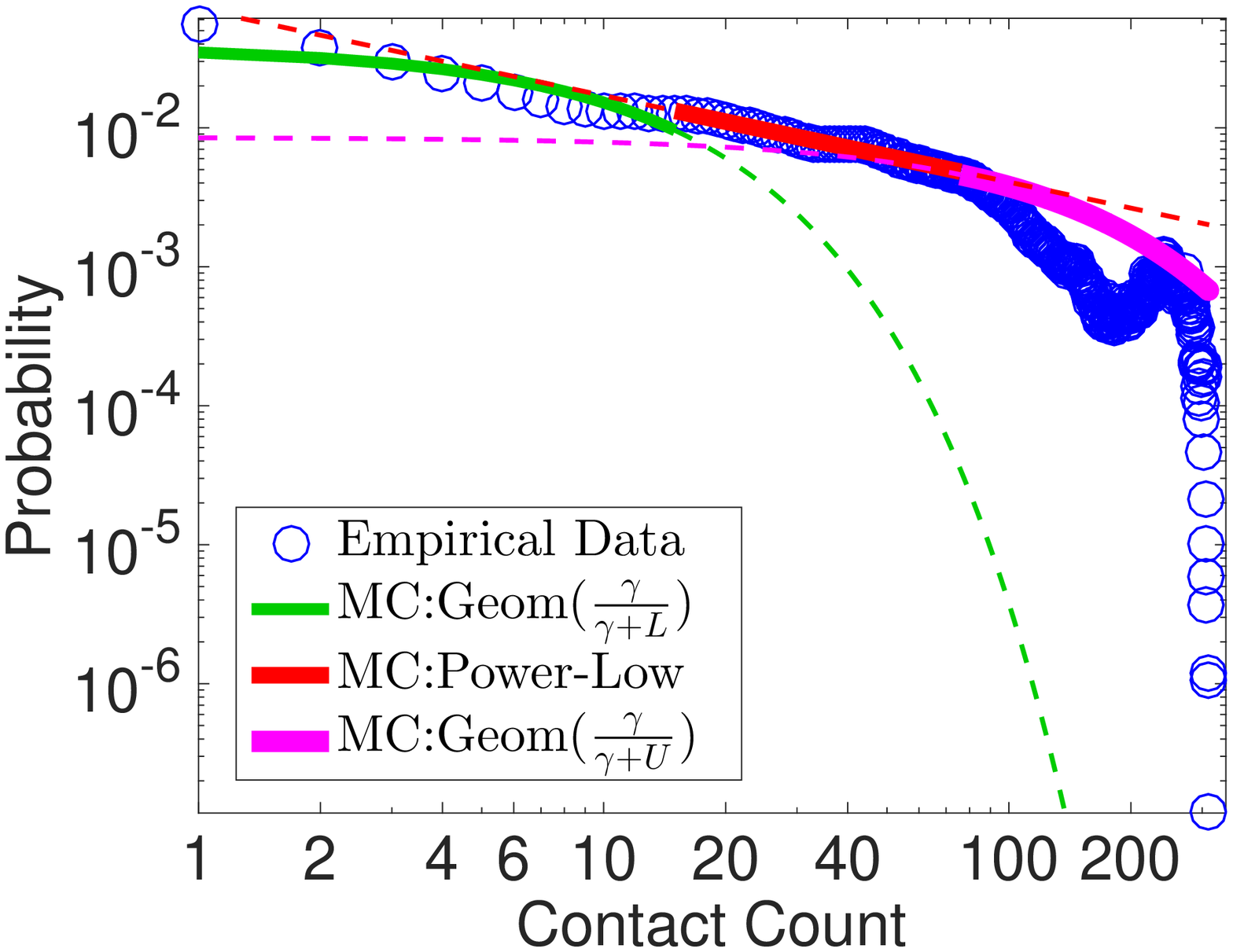}
  }
  \caption{Probability $P(x)$ versus the node degree $x$ on a double logarithmic scale in vehicular networks datasets. The node degree represents the contact counts.}
  \label{fig:v2v}
\end{figure}

\begin{table}[htb]
  \centering
  {\small
    \begin{tabular}{|c|c|c|c|c|c|}
    \hline
      Dataset & $\mathcal{L}$ & $\mathcal{U}$ & exponent & RMSE: ours & RMSE: pl \\\hline
      Rome & $5$ & $11$ & $-1.20$ & $1.7 \times 10^{-3}$ & $0.01$  \\\hline
      Beijing & $4$ & $21$ & $-3.19$ & $9.75 \times 10^{-4}$ & $0.10$ \\\hline
      San Francisco & $15$ & $76$ & $-0.62$ & $9.62 \times 10^{-4}$ & $1.67 \times 10^{-3}$  \\\hline
    \end{tabular}
  }
  \vspace{1pt}
  \caption{Fitting parameters and errors in vehicular networks datasets. $RMSE: ours$ is the fitting error of our model and $RMSE: pl$ is the fitting error by using power-law only.}
  \label{tab:v2v-param}
\end{table}

Fig.~\ref{fig:v2v} shows the distributions of contact counts for Rome, San Francisco, and Beijing taxi datasets. The fitting parameters and errors are shown in Table~\ref{tab:v2v-param}. One can see that the proposed MC model reduces the fitting errors by 42\% - 99\%. One can also observe a similar trend that all the data include three phases. However, the fitting is not as good as that in citation and social networks. This is because, in vehicular networks, vehicles not only {\em build} links with other vehicles (when entering the communication ranges), but also {\em break} some links (when leaving the communication ranges).

\para{Utility: }
The MC model can be used to design a better routing scheme for Delay-Tolerant Networking (DTN), such as vehicular networks. With the intermittent connections in DTN, one of the main challenges is how to select nodes to forward data in order to improve the reachability and throughput. Existing routing schemes rely on exchanging information when two nodes meet, and predicting if the node will get closer to the destination node in the future. These schemes however have two main disadvantages. First, the current prediction is usually based on the exponential or power-law model, which has been shown to have larger prediction errors above. Second, when two nodes meet, they need to exchange the complete or summarized history, which can occupy lots of bandwidths, especially when the intra-connection time is short. With the new MC model, one can calculate the probability directly according to the derived close-form formula, in which vehicles only need to exchange the simple information of their current node degrees. 

\section{Conclusions} 
\label{sec:conclusion}

Based on a new Markov chain (MC) model of a randomly growing network with various regimes of state dynamics characterizing different physical properties of complex networks, this paper establishes a unified framework of several classical complex networks, including Poisson, exponential, and power-law networks.
Significantly, this framework is the first mechanism to investigate the formation mechanism of the trichotomy of observed node-degree density functions from empirical data in many real networks, which has not been addressed in the existing literature. The proposed MC model is capable of offering closed-form expressions of node-degree distributions for all the studied cases. Both simulation and experimental results demonstrate a good match of the proposed model with real datasets, showing its superiority over the classical network models, particularly the power-law network model.

\appendix
\subsection{Residential-time distribution of a node in the basic MC model with external links}
\label{AppendixA:residentialTimeByExtOnly}

In the analysis of the node-degree distribution ${p_k}$, it is very important to know how long (the residential-time $T$) the specified node *
has been existing in the network. Denote the probability density function of the residential-time by ${f_T}\left( t \right)$. At the observation time $\mathcal{T}$, the node-degree distribution ${p_k}$ can be written as follows:
\begin{equation} \label{eqn:node-deg}
{p_k} = {\mathds{E}_{\mathcal{T}}}\left[ {{p_k}\left( t\right)} \right] = \int_0^{\mathcal{T}}
{{p_k}\left( t\right){f_T}\left( t \right)dt}.
\end{equation}

An illustration of the residential-time $T$ of a node is shown by Fig. \ref{fig:residential-time}.

\begin{figure}[htpb!]
  \centering
  \includegraphics[width=0.5\textwidth]{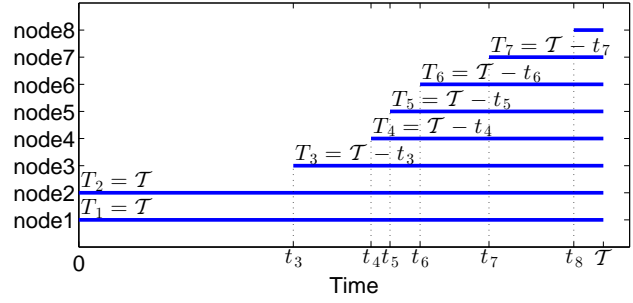}
  \caption{Arrival time of node $i$ is denoted by $t_i$, and residential-time of node $i$ is denoted by $T_i$, which is the same as $\mathcal{T} - t_i$, where $\mathcal{T}$ is the system starting time since the appearance of the first two nodes.}
  \label{fig:residential-time}
\end{figure}

This residential-time distribution ${f_T}\left( t \right)$ of each node in the network, in turn, is related to how the network size $N(t)$ is evolving over time $t$. This evolution of $N(t)$ (or specifically, the differential-difference equation on $p_n(t)$, or the marginal distribution of $N(t)$) can be obtained by summing equation (\ref{eqn:DE-on-2Dpmf}) over $k$, as follows:
\begin{equation} \label{eqn:DE-on-nodesize}
\frac{d}{{dt}}{p_n}(t) = \lambda {S_{n - 1}}{p_{n - 1}}(t)
- \lambda {S_n}{p_n}(t).
\end{equation}
Clearly, equation (\ref{eqn:DE-on-nodesize}) is an inhomogeneous birth process.

Note that the difference between ${S_n}$ and ${S_{n-1}}$ is given by
\[
{S_n} = \sum\limits_{i = 1\hfill\atop
i \ne i'\hfill}^{n - 1} {{{\hat k}_i}}  + \widehat {{k_{i'}} + 1}
+ {\hat k_n} = {S_{n - 1}} + \left( {\widehat {{k_{i'}} + 1} -
\widehat {{k_{i'}}}} \right) + {\hat k_n}
\]
It can be verified that the difference ${S_n} - {S_{n - 1}}$ is
given by $\left( {\widehat {{k_{i'}} +1} - \widehat {{k_{i'}}}}
\right) + {\hat k_n}$. It means that when a new node comes to the
network, it will create a new link and increase the rate of birth in
the process by $\lambda \left( {L +{\varepsilon _{i'}}} \right)$,
where ${\varepsilon _{i'}}$ depends on the existing degree
${k_{i'}}$ of node $i'$. Specifically,
\[
{\varepsilon_{i'}} =
\begin{cases}
0 & {{\text{if }}{k_{i'}} < \mathcal{L}} \\
\mathcal{L} + 1 - L & {{\text{if }}{k_{i'}} = \mathcal{L}} \\
1 & {{\text{if }}\mathcal{L} < {k_{i'}} < \mathcal{U}} \\
U - \mathcal{U} & {{\text{if }}{k_{i'}} = \mathcal{U}} \\
0 & {{\text{if }}{k_{i'}} > \mathcal{U}}
\end{cases}.
\]
It is remarked that the following $\tilde{\varepsilon}_{i'}$ will give the same effect as the above $\varepsilon_{i'}$ (even when $\mathcal{L} \neq L$ or $\mathcal{U} \neq U$) in analyzing the dynamics on $p_n(t)$:
\[
{\tilde{\varepsilon}_{i'}} =
\begin{cases}
0 & {{\text{if }}{k_{i'}} < L} \\
1 & {{\text{if }} L \le {k_{i'}} < U} \\
0 & {{\text{if }}{k_{i'}} \ge U}
\end{cases}.
\]
Specifically, this sum of $\tilde{\varepsilon}_{i'}$ (i.e., $\sum\limits_{i' \text{ is incident by a new node}} {\tilde{\varepsilon}_{i'}}$) gives an upper, but accurate enough, estimate of the sum $\sum\limits_{i' \text{ is incident by a new node}} {\varepsilon_{i'}}$. 

Since this ${\varepsilon _{i'}}$ depends on the current network
size $n$, it can be viewed as a realization of a random variable
${\varepsilon ^n}$, with ${\varepsilon ^0} = 0$ and
${\varepsilon ^1} = 0$, because no other node $i'$ appears in these cases. Hence, expressing $S_n$ and $S_{n-1}$ by $\varepsilon_{i'}$, and replacing the sum of $\varepsilon_{i'}$ by the sum of $\tilde{\varepsilon}_{i'}$, the node dynamics (\ref{eqn:DE-on-nodesize}) satisfy
\[
\frac{d}{{dt}}{p_n}(t) = \lambda \left[ {L\left( {n - 1} \right)
+ \sum\limits_{m = 1}^{n - 1} {{\varepsilon ^{m - 1}}} }
\right]{p_{n - 1}}(t) - \lambda \left[ {Ln + \sum\limits_{m = 1}^n
{{\varepsilon ^{m - 1}}} }\right]{p_n}(t).
\]

Let the final network size at time $\mathcal{T}$ be $N_\mathcal{T}$ and consider several important
cases and their network-size dynamics over time, as follows:

{\it Case} 1. When $\mathcal{L}=1$ and $\mathcal{U}=\infty$, the
node dynamics satisfy
\begin{eqnarray}
\frac{d}{{dt}}{p_n}(t)
&=& \lambda \left[ {\left( {n - 1} \right)
 + \left( {n - 3} \right)} \right]{p_{n - 1}}(t)
 - \lambda \left[ {n + \left( {n - 2} \right)} \right]{p_n}(t)\nonumber\\
&=& 2\lambda \left( {n - 2} \right){p_{n - 1}}(t)
 - 2\lambda \left( {n - 1} \right){p_n}(t).\nonumber
\end{eqnarray}
In this case, the density function ${f_T}\left( t \right)$ of the
residential-time of the specified node * is 
a (scaled) exponential distribution with parameter $2\lambda$, specifically with
probability density function
${f_T}\left( t \right) = \frac{2\lambda {e^{ - 2\lambda t}}}{1 - e^{ - 2\lambda \mathcal{T}}}$.

{\it Case} 2. When $U$ is large (comparable to $N$), the network
size dynamics satisfy
\begin{eqnarray}
\frac{d}{{dt}}{p_n}(t)
&\sim& \lambda \left[ {L\left( {n - 1} \right)
 + \left( {n - 1 - c} \right)} \right]{p_{n - 1}}(t)
 - \lambda \left[ {Ln + \left( {n - c} \right)} \right]{p_n}(t)\nonumber\\
&=& \lambda \left( {\left( {L + 1} \right)
 \left( {n - 1} \right) - \left( {c} \right)} \right){p_{n - 1}}(t)
 - \lambda \left( {\left( {L + 1} \right)n - \left( {c} \right)}
 \right){p_n}(t)\nonumber
\end{eqnarray}
for some $L + 1 \leq c \leq n - 1$.

According to the preferential attachment mechanism, a new node has a higher probability to connect to an existing one with a larger
degree, instead of those with degrees less than or equal to $L$. Hence, $c$ is small compared to $n$, i.e., around $L + 1$, instead around $n - 1$.

In this case, the density function ${f_T}\left( t \right)$ of the
residential-time of the specified node * is given by a (scaled) exponential distribution with parameter close to $\lambda \left(L + 1
\right)$, specifically with probability density function ${f_T}\left( t\right)
= \frac{\lambda \left( {L + 1} \right){e^{ - \lambda \left( {L+ 1} \right)t}}}{1 - e^{ - \lambda \left( {L+ 1} \right)\mathcal{T}}}$.

{\it Case} 3. This is the most practical scenario. When $U \ll N$,
the network size dynamics satisfy
\[
\frac{d}{{dt}}{p_n}(t) \sim \lambda
\left[ {L\left( {n - 1} \right) + c'} \right]{p_{n - 1}}(t)
- \lambda \left[ {Ln + c'} \right]{p_n}(t)
\]
where $c'$ is small, compared to $n$, and satisfies $ 0 \leq c' \leq
n - 1 - \left(L + 1\right)$.

In this case, the density function ${f_T}\left( t \right)$ of the
residential-time of the specified node * is
a (scaled) exponential distribution with parameter close to $\lambda L$, specifically with the
probability density function
${f_T}\left( t \right) = \frac{\lambda L{e^{ - \lambda Lt}}}{1 - \lambda L{e^{ - \lambda L \mathcal{T}}}}$.

It is noted that in all above cases, the denominator of the density function ${f_T}\left( t \right)$ would only result in a difference in the normalization constant in the calculation of $p_k$ in equation (\ref{eqn:node-deg}).


\begin{thebibliography}{10}
\scriptsize
\providecommand{\url}[1]{#1}
\csname url@samestyle\endcsname
\providecommand{\newblock}{\relax}
\providecommand{\bibinfo}[2]{#2}
\providecommand{\BIBentrySTDinterwordspacing}{\spaceskip=0pt\relax}
\providecommand{\BIBentryALTinterwordstretchfactor}{4}
\providecommand{\BIBentryALTinterwordspacing}{\spaceskip=\fontdimen2\font plus
\BIBentryALTinterwordstretchfactor\fontdimen3\font minus
  \fontdimen4\font\relax}
\providecommand{\BIBforeignlanguage}[2]{{%
\expandafter\ifx\csname l@#1\endcsname\relax
\typeout{** WARNING: IEEEtran.bst: No hyphenation pattern has been}%
\typeout{** loaded for the language `#1'. Using the pattern for}%
\typeout{** the default language instead.}%
\else
\language=\csname l@#1\endcsname
\fi
#2}}
\providecommand{\BIBdecl}{\relax}
\BIBdecl

\bibitem{barabasi:science:random}
\BIBentryALTinterwordspacing
A.-L. Barabasi and R.~Albert, ``Emergence of scaling in random networks,''
  \emph{Science}, vol. 286, no. 5439, pp. 509--512, 1999. [Online]. Available:
  \url{http://www.sciencemag.org/cgi/content/abstract/286/5439/509}
\BIBentrySTDinterwordspacing

\bibitem{jeong:nature:metabolic}
\BIBentryALTinterwordspacing
H.~Jeong, B.~Tombor, R.~Albert, Z.~N. Oltvai, and A.~L. Barab\'{a}si, ``The
  large-scale organization of metabolic networks,'' \emph{Nature}, vol. 407,
  no. 6804, pp. 651--654, 10 2000. [Online]. Available:
  \url{http://dx.doi.org/10.1038/35036627}
\BIBentrySTDinterwordspacing

\bibitem{martin:pnas:powerlaw}
\BIBentryALTinterwordspacing
H.~García~Martín and N.~Goldenfeld, ``On the origin and robustness of
  power-law species–area relationships in ecology,'' \emph{Proceedings of the
  National Academy of Sciences}, vol. 103, no.~27, pp. 10\,310--10\,315, 2006.
  [Online]. Available: \url{http://www.pnas.org/content/103/27/10310.abstract}
\BIBentrySTDinterwordspacing

\bibitem{price:jasis:citation}
D.~D.~S. Price, ``A general theory of bibliometric and other cumulative
  advantage processes,'' \emph{Journal of the American Society for Information
  Science}, pp. 292--306, 1976.

\bibitem{martin:pre:citation}
\BIBentryALTinterwordspacing
T.~Martin, B.~Ball, B.~Karrer, and M.~E.~J. Newman, ``Coauthorship and citation
  patterns in the physical review,'' \emph{Phys. Rev. E}, vol.~88, p. 012814,
  Jul 2013. [Online]. Available:
  \url{http://link.aps.org/doi/10.1103/PhysRevE.88.012814}
\BIBentrySTDinterwordspacing

\bibitem{ErdoRenyi:pm:rg}
P.~Erd{\"{o}}s and A.~R{\'{e}}nyi, ``On random graphs i,'' \emph{Publicationes
  Mathematicae}, vol.~6, p. 290–297, 1959.

\bibitem{callaway:pre:rgg}
D.~S. Callaway, J.~E. Hopcroft, J.~M. Kleinberg, M.~E.~J. Newman, and S.~H.
  Strogatz, ``Are randomly grown graphs really random,'' \emph{Physical Review
  E}, vol.~64, no.~4, 2001, 041902.

\bibitem{liu:pla:pa}
\BIBentryALTinterwordspacing
Z.~Liu, Y.-C. Lai, N.~Ye, and P.~Dasgupta, ``Connectivity distribution and
  attack tolerance of general networks with both preferential and random
  attachments,'' \emph{Physics Letters A}, vol. 303, no. 5–6, pp. 337 -- 344,
  2002. [Online]. Available:
  \url{http://www.sciencedirect.com/science/article/pii/S0375960102013178}
\BIBentrySTDinterwordspacing

\bibitem{yule:ptrsb:math}
G.~U. { Yule}, ``{A mathematical theory of evolution, based on the conclusions
  of Dr. J. C. Willis, F.R.S.}'' \emph{Philosophical Transactions of the Royal
  Society of London Series B}, vol. 213, pp. 21--87, 1925.

\bibitem{bianconi:prl:complexnetworks}
\BIBentryALTinterwordspacing
G.~Bianconi and A.~L. Barab\'{a}si, ``{Bose-Einstein condensation in complex
  networks},'' \emph{Physical Review Letters}, vol.~86, pp. 5632--5635, 2001.
  [Online]. Available: \url{http://link.aps.org/abstract/PRL/v86/p5632}
\BIBentrySTDinterwordspacing

\bibitem{krapivsky:prl:random}
\BIBentryALTinterwordspacing
P.~L. Krapivsky, S.~Redner, and F.~Leyvraz, ``Connectivity of growing random
  networks,'' \emph{Phys. Rev. Lett.}, vol.~85, pp. 4629--4632, Nov 2000.
  [Online]. Available:
  \url{http://link.aps.org/doi/10.1103/PhysRevLett.85.4629}
\BIBentrySTDinterwordspacing

\bibitem{bianconi:epl:evolving_networks}
\BIBentryALTinterwordspacing
G.~Bianconi and A.-L. Barabási, ``Competition and multiscaling in evolving
  networks,'' \emph{EPL (Europhysics Letters)}, vol.~54, no.~4, p. 436, 2001.
  [Online]. Available: \url{http://stacks.iop.org/0295-5075/54/i=4/a=436}
\BIBentrySTDinterwordspacing

\bibitem{medo:prl:powerlaw_decay}
M.~{Medo}, G.~{Cimini}, and S.~{Gualdi}, ``{Temporal effects in the growth of
  networks},'' \emph{Physical Review Letters}, vol. 107, no.~23, p. 238701,
  Dec. 2011.

\bibitem{albert:prl:evolving_networks}
\BIBentryALTinterwordspacing
R.~Albert and A.-L. Barab\'asi, ``Topology of evolving networks: Local events
  and universality,'' \emph{Phys. Rev. Lett.}, vol.~85, pp. 5234--5237, Dec
  2000. [Online]. Available:
  \url{http://link.aps.org/doi/10.1103/PhysRevLett.85.5234}
\BIBentrySTDinterwordspacing

\bibitem{tang:kdd08:citation}
J.~Tang, J.~Zhang, L.~Yao, J.~Li, L.~Zhang, and Z.~Su, ``Arnetminer: Extraction
  and mining of academic social networks,'' in \emph{KDD'08}, 2008, pp.
  990--998.

\bibitem{aps-cit}
``The {American Physical Society} data sets for research,''
  \url{http://journals.aps.org/datasets}, 2014.

\bibitem{leskovec:kdd05:us-patent-cit}
\BIBentryALTinterwordspacing
J.~Leskovec, J.~Kleinberg, and C.~Faloutsos, ``Graphs over time: Densification
  laws, shrinking diameters and possible explanations,'' in \emph{Proceedings
  of the Eleventh ACM SIGKDD International Conference on Knowledge Discovery in
  Data Mining}, ser. KDD '05.\hskip 1em plus 0.5em minus 0.4em\relax New York,
  NY, USA: ACM, 2005, pp. 177--187. [Online]. Available:
  \url{http://doi.acm.org/10.1145/1081870.1081893}
\BIBentrySTDinterwordspacing


\bibitem{ErdoRenyi:pm:rg}
P.~Erd{\"{o}}s and A.~R{\'{e}}nyi, ``On random graphs i,'' \emph{Publicationes
  Mathematicae}, vol.~6, p. 290–297, 1959.

\bibitem{yule:ptrsb:math}
G.~U. { Yule}, ``{A mathematical theory of evolution, based on the conclusions
  of Dr. J. C. Willis, F.R.S.}'' \emph{Philosophical Transactions of the Royal
  Society of London Series B}, vol. 213, pp. 21--87, 1925.

\bibitem{callaway:pre:rgg}
D.~S. Callaway, J.~E. Hopcroft, J.~M. Kleinberg, M.~E.~J. Newman, and S.~H.
  Strogatz, ``Are randomly grown graphs really random,'' \emph{Physical Review
  E}, vol.~64, no.~4, 2001, 041902.

\bibitem{hauer:ps90:prony}
J.~Hauer, C.~Demeure, and L.~Scharf, ``Initial results in prony analysis of
  power system response signals,'' \emph{Power Systems, IEEE Transactions on},
  vol.~5, no.~1, pp. 80 --89, Feb. 1990.

\bibitem{tang:kdd08:citation}
J.~Tang, J.~Zhang, L.~Yao, J.~Li, L.~Zhang, and Z.~Su, ``Arnetminer: Extraction
  and mining of academic social networks,'' in \emph{KDD'08}, 2008, pp.
  990--998.

\bibitem{aps-cit}
``The {American Physical Society} data sets for research,''
  http://journals.aps.org/datasets, 2014.

\bibitem{leskovec:kdd05:us-patent-cit}
\BIBentryALTinterwordspacing
J.~Leskovec, J.~Kleinberg, and C.~Faloutsos, ``Graphs over time: Densification
  laws, shrinking diameters and possible explanations,'' in \emph{Proceedings
  of the Eleventh ACM SIGKDD International Conference on Knowledge Discovery in
  Data Mining}, ser. KDD '05.\hskip 1em plus 0.5em minus 0.4em\relax New York,
  NY, USA: ACM, 2005, pp. 177--187. [Online]. Available:
  \url{http://doi.acm.org/10.1145/1081870.1081893}
\BIBentrySTDinterwordspacing

\bibitem{snapnets}
J.~Leskovec and A.~Krevl, ``{SNAP Datasets}: {Stanford} large network dataset
  collection,'' \url{http://snap.stanford.edu/data}, Jun. 2014.

\bibitem{mcauley:tkdd14:social}
\BIBentryALTinterwordspacing
J.~Mcauley and J.~Leskovec, ``Discovering social circles in ego networks,''
  \emph{ACM Trans. Knowl. Discov. Data}, vol.~8, no.~1, pp. 4:1--4:28, Feb.
  2014. [Online]. Available: \url{http://doi.acm.org/10.1145/2556612}
\BIBentrySTDinterwordspacing

\bibitem{leskovec:nips12:tiwtter}
\BIBentryALTinterwordspacing
J.~Leskovec and J.~J. Mcauley, ``Learning to discover social circles in ego
  networks,'' in \emph{Advances in Neural Information Processing Systems 25},
  F.~Pereira, C.~Burges, L.~Bottou, and K.~Weinberger, Eds.\hskip 1em plus
  0.5em minus 0.4em\relax Curran Associates, Inc., 2012, pp. 539--547.
  [Online]. Available:
  \url{http://papers.nips.cc/paper/4532-learning-to-discover-social-circles-in-ego-networks.pdf}
\BIBentrySTDinterwordspacing

\bibitem{rome-taxi}
L.~Bracciale, M.~Bonola, P.~Loreti, G.~Bianchi, R.~Amici, and A.~Rabuffi,
  ``{CRAWDAD} dataset roma/taxi (v. 2014-07-17),'' Downloaded from
  http://crawdad.org/roma/taxi/20140717, Jul. 2014.

\bibitem{yuan:gis10:beijing-taxi}
\BIBentryALTinterwordspacing
J.~Yuan, Y.~Zheng, C.~Zhang, W.~Xie, X.~Xie, G.~Sun, and Y.~Huang, ``T-drive:
  Driving directions based on taxi trajectories,'' in \emph{Proceedings of the
  18th SIGSPATIAL International Conference on Advances in Geographic
  Information Systems}, ser. GIS '10.\hskip 1em plus 0.5em minus 0.4em\relax
  New York, NY, USA: ACM, 2010, pp. 99--108. [Online]. Available:
  \url{http://doi.acm.org/10.1145/1869790.1869807}
\BIBentrySTDinterwordspacing

\bibitem{yuan:kdd11:beijing-taxi}
\BIBentryALTinterwordspacing
J.~Yuan, Y.~Zheng, X.~Xie, and G.~Sun, ``Driving with knowledge from the
  physical world,'' in \emph{Proceedings of the 17th ACM SIGKDD International
  Conference on Knowledge Discovery and Data Mining}, ser. KDD '11.\hskip 1em
  plus 0.5em minus 0.4em\relax New York, NY, USA: ACM, 2011, pp. 316--324.
  [Online]. Available: \url{http://doi.acm.org/10.1145/2020408.2020462}
\BIBentrySTDinterwordspacing

\bibitem{sf-taxi}
\BIBentryALTinterwordspacing
M.~Piorkowski, N.~Sarafijanovic-Djukic, and M.~Grossglauser, ``{CRAWDAD}
  dataset epfl/mobility (v. 2009-02-24),'' Downloaded from
  http://crawdad.org/epfl/mobility/20090224, Feb. 2009.
\BIBentrySTDinterwordspacing
\end{thebibliography}
%

\end{document}